\documentclass[12pt]{article}

\usepackage[round, authoryear]{natbib}
\usepackage{mathrsfs}
\usepackage[psamsfonts]{amssymb}
\usepackage{natbib,graphicx,setspace,lscape,longtable}

\usepackage{epsfig}
\usepackage{mathrsfs,amsmath,amsthm,color,url}
\usepackage{verbatim}
\usepackage{bm}
\makeatletter
\def\singlespace{\def\baselinestretch{1}\@normalsize}

\usepackage[colorlinks,citecolor=blue,urlcolor=blue]{hyperref}
\usepackage{breakurl}

\UrlBigBreaks

\bibpunct{(}{)}{;}{a}{,}{,}

\newtheorem{theorem}{Theorem}
\newtheorem{lemma}{Lemma}
\newtheorem{proposition}{Proposition}

\newtheorem{corollary}{Corollary} 

\theoremstyle{definition}

\def\beq{\begin{equation}}
\def\eeq{\end{equation}}
\def\beqr{\begin{eqnarray}}
\def\eeqr{\end{eqnarray}}
\def\beqrs{\begin{eqnarray*}}
\def\eeqrs{\end{eqnarray*}}
\def\bet{\begin{theorem}}
\def\eet{\end{theorem}}
\def\bel{\begin{lemma}}
\def\eel{\end{lemma}}
\def\bep{\begin{proposition}}
\def\eep{\end{proposition}}
\def\bg{\begin{figure}[tbph]\begin{center}}
\def\eg{\end{center}\end{figure}}

\def\bc{\begin{center}}
\def\ec{\end{center}}

\renewcommand{\baselinestretch}{1.4}
\numberwithin{equation}{section}

\begin{document}

\title{Corrected Bayesian information criterion for stochastic block models
}
\author{{\sc  Jianwei Hu$^a$, Hong Qin$^{a,b}$, Ting Yan$^a$ and Yunpeng Zhao$^c$}\\
$^a$Central China Normal University, $^b$Zhongnan University of Economics and Law\\
 \& $^c$Arizona State University}
\date{\today}
\maketitle{}

\begin{singlespace}
\begin{footnotetext}[1]
{Jianwei Hu, Department of Statistics, Central China Normal University, Wuhan, P.R. China, 430079 (Email: jwhu@mail.ccnu.edu.cn).
Hong Qin, Department of Statistics, Central China Normal University, Wuhan, P.R. China, 430079 (Email: qinhong@mail.ccnu.edu.cn) and Department of Statistics, Zhongnan University of Economics and Law, Wuhan, P.R. China, 430073 (Email: z0000020@zuel.edu.cn).
Ting Yan, Department of Statistics, Central China Normal University, Wuhan, P.R. China, 430079 (Email: tingyanty@mail.ccnu.edu.cn).
Yunpeng Zhao, School of Mathematical and Natural Sciences, Arizona State University, Tempe AZ, USA, 85281 (Email: Yunpeng.Zhao@asu.edu). All four authors contributed equally to this article.
}
\end{footnotetext}
\end{singlespace}

\begin{abstract}
Estimating the number of communities is one of the fundamental problems in community detection.
We re-examine the Bayesian paradigm for stochastic block models and propose a ``corrected Bayesian information criterion" (CBIC),
to determine the number of communities
and show that the proposed criterion is consistent under mild conditions as the size of the network and the number of communities go to infinity.
The CBIC outperforms those used in \cite{Wang:Bickel:2017} and \cite{Saldana:Yu:Feng:2017} which tend
to underestimate and overestimate the number of communities, respectively.
The results are further extended to degree corrected stochastic block models. Numerical studies demonstrate our theoretical results.\\

\noindent{\bf KEY WORDS:} Consistency; degree corrected stochastic block model; network data; stochastic block model.

\end{abstract}
\renewcommand{\baselinestretch}{1.75}
\baselineskip=24pt

\section{Introduction}

Community structure is one of the most widely used structures for network data.
For instance, peoples form groups in social networks based on common locations, interests, occupations, etc.;
proteins form communities based on functions in metabolic networks; publications can be grouped to communities by research topics in
citation networks.
The links (or edges) between nodes are generally dense within communities and relatively sparse between communities.
Identifying such sub-groups provides important insights into network formation mechanism and
how network topology affects each other.

The stochastic block model (SBM) proposed by \cite{Holland:Laskey:Leinhardt:1983} is one of the best studied network models for community
structures. See \cite{Snijders:Nowicki:1997} and \cite{Nowicki:Snijders:2001} for a first application of the SBM in community detection. We briefly describe the model.
Let $A\in \{0,1\}^{n\times n}$ be the symmetric adjacency matrix of an undirected graph with $n$ nodes.
In the SBM with $k$ communities, each node is associated with a community, labeled by $z_k(i)$,
where $z_k(i)\in [k]$. Here $[m]=\{1,\ldots,m\}$ for any positive integer $m$. In other words, the nodes are given a
community assignment $z_k: [n]\rightarrow[k]^n$.
The diagonal entries of  $A$ are zeros (no self-loop) and the entries of the upper triangle matrix $A$ are independent Bernoulli random variables
with success probabilities that only depend on the community labels of nodes $i$ and $j$.
That is, all edges are independently generated given the node communities, and for
a certain probability matrix $\mathbf{\theta}_k=\{\theta_{kab}\}_{1\leq a,b\leq k}$,
\[
P(A_{ij}=1\mid z_k(i),z_k(j))=\theta_{kz_k(i)z_k(j)}.
\]
For simplicity, $\theta_k$ and $z_k$ are abbreviated to $\theta$ and $z$, respectively.

A wide variety of methods have been proposed to estimate the latent community membership of nodes in an SBM,
including modularity \citep{Newman:2006a},
profile-likelihood \citep{Bickel:Chen:2009}, pseudo-likelihood \citep{Amini:etal:2013},
variational methods \citep{Daudin:Picard:Robin:2008, Latouche:Birmele:Amroise:2012},
spectral clustering \citep{Rohe:Chatterjee:Yu:2011, Fishkind:etal:2013, Jin:2015}, belief propagation \citep{Decelle:etal:2011}, etc.
The asymptotic properties of these methods have also been established under different settings
\citep{Bickel:Chen:2009, Rohe:Chatterjee:Yu:2011, Celisse:Daudin:Pierre:2012, Bickel:etal:2013, Gao:Lu:Zhou:2015, Zhang:Zhou:2016}.

However, most of these works assume that the number of communities $k$ is known a priori.
In a real-world network, $k$ is usually unknown and needs to be estimated.
Therefore, it is of importance to investigate how to choose $k$ (called model selection in this article). Some methods have been proposed in recent years,
including a recursive approach \citep{Zhao:Levina:Zhu:2011}, spectral methods \citep{Le:Levina:2015}, sequential tests \citep{Bickel:Sarkar:2015, Lei:2016},
and network cross-validation \citep{Chen:Lei:2018}.
The likelihood-based methods for model selection have also been proposed
\citep{Daudin:Picard:Robin:2008, Latouche:Birmele:Amroise:2012, Saldana:Yu:Feng:2017}.

\cite{Wang:Bickel:2017} proposed a penalized likelihood method with the penalty  function
\begin{equation}\label{eq:penalized:wb}
\lambda \frac{k(k+1)}{2} n\log n,
\end{equation}
where $\lambda$ is a tuning parameter.
An alternative penalty function $\frac{k(k+1)}{2}\log n$ (called the ``BIC")
was used to select the number of communities in \cite{Saldana:Yu:Feng:2017}.
As will be shown later in the paper, using the penalty function \eqref{eq:penalized:wb} and the BIC to estimate $k$ tends to underestimate and overestimate the number of communities respectively. We therefore propose a ``corrected Bayesian information criterion'' (CBIC) that is in the midway of those two criteria.
Specifically, we propose the following penalty function
\begin{equation}\label{eq:lambda}
\lambda n\log k+\frac{k(k+1)}{2}\log n,
\end{equation}
which is lighter than \eqref{eq:penalized:wb} used by \cite{Wang:Bickel:2017} and is heavier than  the BIC penalty used by \cite{Saldana:Yu:Feng:2017}. Rigorously speaking, \cite{Wang:Bickel:2017} dealt with the marginal log-likelihood where $z$ as latent variables are integrated out, while we plug  a single estimated community assignment into the log-likelihood. 

For fixed $k$, \cite{Wang:Bickel:2017} established the limiting distribution of the log-likelihood ratio
under model misspecification -- both underfitting and overfitting, and thereby determined an upper bound $o(n^2)$ and a lower bound $n$ of the order of the penalty term for a consistent model selection. Based on the work of \cite{Wang:Bickel:2017}, we derived new upper and lower bounds for increasing $k$. According to our theory (see the proof of Theorem 4 for details), the main orders of both the upper and lower bounds are $n\log k$. In this sense, the bounds we obtained are sharp. Based on these results, we establish the consistency of the CBIC in determining the number of communities. The results are further extended to degree corrected block models. Along the way of proving consistency, we also obtain the Wilks theorem for stochastic block models.

The remainder of the paper is organized as follows.
In Section \ref{section:corrected BIC}, we re-examine the Bayesian paradigm for stochastic block models and propose the CBIC to determine the number of communities. In Section \ref{section:2}, we analyze the asymptotic behavior of the log-likelihood ratio and establish
their asymptotic distributions.
In Section \ref{section:modelselection},
we establish the consistency of the estimator for the number of communities. We extend our results to degree corrected stochastic block models in Section \ref{section:extension}.
The numerical studies are given in Section \ref{section:experiments}.
Some further discussions are made in Section \ref{section:discussion}.
All proofs are given in the Appendix.

\section{Corrected BIC}
\label{section:corrected BIC} In this section, we re-examine the Bayesian paradigm for the SBM and propose a corrected family of Bayesian information criteria.

For any fixed $(\theta,z)$, the log-likelihood
of the adjacency matrix $A$ under the stochastic block model is
\[
\log f(A|\theta,z)=\sum_{1\leq a\leq b\leq k}(m_{ab}\log\theta_{ab}+(n_{ab}-m_{ab})\log(1-\theta_{ab})),
\]
where $n_a=\sum_{i=1}^n\mathbf{1}\{z(i)=a\}$, for $a\neq b$,
$$n_{ab}=n_an_b,\,\,m_{ab}=\sum_{i=1}^n\sum_{j\neq i}A_{ij}\mathbf{1}\{z(i)=a,z(j)=b\},$$
$$n_{aa}=n_a(n_a-1)/2,\,\,m_{aa}=\sum_{i<j}A_{ij}\mathbf{1}\{z(i)=a,z(j)=a\}.$$

\cite{Saldana:Yu:Feng:2017} used the following penalized likelihood function to select the optimal number of communities:
\begin{equation}
\label{eq:bic}
\check{\ell}(k)=\sup_{\theta\in\Theta_{k}}\log f(A|\theta,z)-\frac{k(k+1)}{2}\log n,
\end{equation}
where $\Theta_{k}=[0,1]^{\frac{k(k+1)}{2}}$. Note that \eqref{eq:bic} is not a standard BIC criterion but a BIC approximation of the log-likelihood for given $z$ (see \eqref{eq:bicapprox}). \cite{Saldana:Yu:Feng:2017} essentially estimates the number of communities $k$ using the following criterion, which we refer to as the BIC hereafter:
\[
\bar{\ell}(k)=\max_{z\in [k]^n}\sup_{\theta\in\Theta_{k}}\log f(A|\theta,z)-\frac{k(k+1)}{2}\log n.
\] According to our simulation studies,
this BIC tends to overestimate the number on communities (see Section \ref{section:experiments}). We now provide some insight of why this phenomenon occurs.

Let $Z$ be the set of all possible community assignments under consideration and let  $\xi(z)$ be a prior probability of community assignment $z$. Assume that the prior density of $\theta$ is given by $\pi(\theta)$. Then the posterior probability of $z$ is
\[
P(z|A)=\frac{g(A|z)\xi(z)}{\sum_{z\in Z}g(A|z)\xi(z)},
\]
where $g(A|z)$ is the likelihood of community assignment $z$, given by
\[
g(A|z)=\int f(A|\theta,z)\pi(\theta)d\theta.
\]

Under the Bayesian paradigm, a community assignment $\hat{z}$ that maximizes the posterior probability is selected. Since $\sum_{z\in Z}g(A|z)\xi(z)$ is a constant, \[\hat{z}=\max_{z\in [k]^n}g(A|z)\xi(z).\]
By using a BIC approximation\footnote{The BIC approximation is a general principle and is not to be confused with the BIC criterion used in \cite{Saldana:Yu:Feng:2017}.} \citep{Schwarz:1978,Saldana:Yu:Feng:2017}, we have
\begin{equation}\label{eq:bicapprox}
\log (\int f(A|\theta,z)\pi(\theta)d\theta)=\sup_{\theta\in\Theta_{k}}\log f(A|\theta,z)-\frac{1}{2}\frac{k(k+1)}{2}\log \frac{n(n-1)}{2}+O(1).
\end{equation}
Thus,
\begin{equation}\label{eq:cbic}
\log g(A|z)\xi(z)=\sup_{\theta\in\Theta_{k}}\log f(A|\theta,z)-\frac{k(k+1)}{2}\log n+O(1)+\log\xi(z).
\end{equation}

By comparing equations \eqref{eq:bic} and \eqref{eq:cbic},
we can see that the BIC essentially assumes that $\xi(z)$ is a constant for $z$ over $Z$, i.e., $\xi(z)=1/\tau(Z)$, where $\tau(Z)$ is the size of $Z$.
Suppose that the number of nodes in the network is $n=500$. The set of community assignments for $k=2$, $Z_2$, has size $2^{500}$, while the set of community assignments for $k=3$, $Z_3$, has size $3^{500}$. The constant prior in the BIC assigns probabilities to $Z_k$ proportional to their sizes. Thus the probability assigned to $Z_3$ is $1.5^{500}$ times that assigned to $Z_2$. Community assignments with a larger number of communities get much higher probabilities than community assignments with fewer communities. This provides an explanation for why
the BIC tends to overestimate the number of communities.

This re-examination of the BIC naturally motivates us to consider a new prior over $Z$. Assume that  $Z$ is partitioned into $\bigcup_{k=1}Z_k$. Let $\tau(Z_k)$ be the size of $Z_k$. We assign the prior distribution over $Z$ in the following manner.
We assign an equal probability to $z$ in the same $Z_k$, i.e., $P(z|Z_k)=1/\tau(Z_k)$ for any $z\in Z_k$.
This is due to that all the community assignments in $Z_k$ are equally plausible.
Next, instead of assigning probabilities $P(Z_k)$ proportional to $\tau(Z_k)$ , we assign $P(Z_k)$ proportional to $\tau^{-\delta}(Z_k)$ for some $\delta$. Here $\delta>0$ implies that a small number of communities are plausible while $\delta<0$ implies that a large number of communities are plausible. This results in the prior probability
\[
\xi(z)=P(z|Z_k)P(Z_k)\varpropto\tau^{-\lambda}(Z_k), ~~ z\in Z_k,
\]
where $\lambda=1+\delta$.
Thus,
\[
\log g(A|z)\xi(z)=\sup_{\theta\in\Theta_{k}}\log f(A|\theta,z)-\frac{k(k+1)}{2}\log n+O(1)-\lambda n\log k.
\]
This type of prior distribution on the community assignment suggests a corrected BIC criterion (CBIC) as follows:
\begin{equation}
\label{eq:penalized:cbic}
\ell(k)=\max_{z\in [k]^n}\sup_{\theta\in\Theta_{k}}\log f(A|\theta,z)-\left[ \lambda n\log k +\frac{k(k+1)}{2}\log n \right],
\end{equation}
where the second term is the penalty and $\lambda\geq0$ is a tuning parameter.
Then we estimate $k$ by maximizing the penalized likelihood function:
\[
\hat{k}=\arg\max_{k}\ell(k).
\]

We make some remarks on the choice of the tuning parameter.
If we have no prior information on the number of communities -- i.e.,
both small number of communities and large number of communities are equally plausible,
then $\lambda=1$ ($\delta=0$) is a good choice. It is similar to the case of variable selection in regression analysis, where the BIC is also tuning free.

The CBIC is also related to an integrated classification likelihood (ICL) method proposed by \cite{Daudin:Picard:Robin:2008}. The penalty function in ICL can be written as
\begin{equation}\label{eq:penalized:dpr}
\sum_{a=1}^kn_a\log(n/n_a)+\frac{k^2+2k}{2}\log n.
\end{equation}
The penalty term in the ICL criterion uses unknown quantities $n_a$ that need to be estimated, and is thus not a standard BIC-type criterion. With equal-sized estimated communities, this penalty is almost the same as the CBIC with $\lambda=1$ since $\sum_{a=1}^kn_a\log(n/n_a) = n \log k $. However, the CBIC has tuning parameter $\lambda$ that gives more flexibility. If we have prior information that large numbers of communities
 are plausible, $\lambda<1$ ($\delta<0$) is a good choice and the CBIC performs significantly better than the ICL in simulation studies (see Section \ref{section:experiments}).

 Moreover, theoretical properties of ICL have not been well-studied while the consistency of the CBIC will be established in this paper. In order to obtain the consistency of the CBIC, we analyze the asymptotic order of the log-likelihood ratio under model-misspecification in the next section.

\section{Asymptotics of the log-likelihood ratio}
\label{section:2}

In this section, we present the order of the log-likelihood ratio built on the work of \cite{Wang:Bickel:2017}.
The results here will be used for the proof of Theorem \ref{theorem:modelselection:a} in the next section.

We consider the following log-likelihood ratio
\[
L_{k,k'}=\max_{z\in [k']^n}\sup_{\theta\in\Theta_{k'}}\log f(A|\theta,z)-\log f(A|\theta^*,z^*),
\]
where $\theta^*$ and $z^*$ are the true parameters. Further, $k'$ is the number of communities under the alternative model and $k$ is the true
number of communities. Therefore, the comparison is made between the correct $k$-block model and a fitted $k'$-block model.

The asymptotic distributions of $L_{k,k'}$ for the cases $k'<k$ and $k'>k$ are given in this section.
For the case $k'=k$, we establish the Wilks theorem.

\subsection{$k'<k$}
We start with $k'=k-1$.
As discussed in \cite{Wang:Bickel:2017}, a $(k-1)$-block model can be obtained by merging blocks in a $k$-block model.
Specifically, given the true labels $z^*\in[k]^n$ and $p=(p_{ab})_{k\times k}$, where $p_{ab}=n_{ab}(z^*)/(\frac{n(n-1)}{2})$,
we define a merging operation $U_{a,b}(\theta^*,p)$ which combines blocks $a$ and $b$ in $\theta^*$
by taking weighted averages with proportions in $p$. For example, for $\theta'=U_{k-1,k}(\theta^*,p)$,
$$\theta'_{ab}=\theta_{ab}^*\,\,\, \mathrm{for}\,\,\, 1\leq a\leq b\leq k-2,$$
$$\theta'_{a(k-1)}=\frac{p_{a(k-1)}\theta_{a(k-1)}^*+p_{ak}\theta_{ak}^*}{p_{a(k-1)}+p_{ak}}\,\,\, \mathrm{for}\,\,\, 1\leq a\leq k-2,$$
$$\theta'_{(k-1)(k-1)}=\frac{p_{(k-1)(k-1)}\theta_{(k-1)(k-1)}^*+p_{(k-1)k}\theta_{(k-1)k}^*+p_{kk}\theta_{kk}^*}{p_{(k-1)(k-1)}+p_{(k-1)k}+p_{kk}}.$$
For consistency, when merging two blocks $(a,b)$ with $b>a$, the new merged block will be relabeled as $a$ and all the blocks $c$ with $c>b$
will be relabeled as $c-1$. Using this scheme, we also obtain the merged node labels $U_{a,b}(z^*)$. For $z'=U_{k-1,k}(z^*)$, define
$$m_{ab}(z')=m_{ab},\,\,\,n_{ab}(z')=n_{ab}\,\,\, \mathrm{for}\,\,\, 1\leq a\leq b\leq k-2,$$
$$m_{a(k-1)}(z')=m_{a(k-1)}+m_{ak},\,\,\,n_{a(k-1)}(z')=n_{a(k-1)}+n_{ak}\,\,\, \mathrm{for}\,\,\, 1\leq a\leq k-2,$$
$$m_{(k-1)(k-1)}(z')=m_{(k-1)(k-1)}+m_{(k-1)k}+m_{kk},\,\,\,n_{(k-1)(k-1)}(z')=n_{(k-1)(k-1)}+n_{(k-1)k}+n_{kk}.$$
To obtain the asymptotic distribution of $L_{k,k'}$, we
need the following conditions.\\
(A1) There exists $C_1>0$ such that $\min_{1\leq a\leq k}n_a\geq C_1n/k$ for all $n$.\\
(A2) Any two rows of $\theta^*$ should be distinct.\\
(A3) The entries of $\theta^*$ are uniformly bounded away from 0 and 1.

In Condition (A1), the lower bound on the smallest community
size requires that the size of each community is at least proportional to
$n/k$. This is a reasonable and mild condition; for example, it is satisfied almost surely if the membership vector
 is generated from a multinomial distribution with $n$ trials and probability
$\pi=(\pi_1, \ldots, \pi_{k})$ such that $\min_{1\le u \le k} \pi_u \ge C_1/k$.  This condition was also used in \cite{Lei:2016}. Condition (A2) requires that the merged model cannot be collapsed further to a smaller model.

Condition (A3) requires the overall density of the network to be a constant. To allow for a sparser network, we can further parametrize $\theta^*=\rho_n\tilde{\theta}^*$ where $\tilde{\theta}^*$ is a constant and $\rho_n\rightarrow 0$ at the rate $n \rho_n/\log n \rightarrow \infty$. (Using this parametrization $\rho_n \equiv 1$ indicates  a constant graph density.) Condition (A3) in this case becomes

(A$3'$) The entries of $\tilde{\theta}^*$ are uniformly bounded away from 0 and 1.

The asymptotic distribution of $L_{k,k-1}$ for a dense network is stated below, the proof of which is given in the Appendix. 

\begin{theorem}\label{theorem:underfit:a}
Suppose that $A\sim P_{\theta^*, z^*}$, conditions (A1)-(A3) hold, and $\rho_n \equiv 1$.  If $k=o((n/\log n)^{1/2})$, we have
\[
(n^{-1}L_{k,k-1}-n\mu)/\sigma(\theta^*)  \stackrel{d}{\rightarrow}  N(0,1),
\]
where
\[
\mu = \frac{1}{n^2}(\sum_{k-1\leq a\leq b\leq k}n_{ab}'(\theta_{ab}'\log\frac{\theta_{ab}'}{1-\theta_{ab}'}+\log(1-\theta_{ab}'))
-\sum_{k-1\leq a\leq b\leq k}n_{ab}(\theta_{ab}^*\log\frac{\theta_{ab}^*}{1-\theta_{ab}^*}+\log(1-\theta_{ab}^*))),
\]
\[
\sigma^2(\theta^*) = \frac{1}{n^2}(\sum_{k-1\leq a\leq b\leq k}n'_{ab}\theta_{ab}'(1-\theta_{ab}')
(\log\frac{\theta_{ab}'}{1-\theta_{ab}'})^2+\sum_{k-1\leq a\leq b\leq k}n_{ab}\theta_{ab}^*(1-\theta_{ab}^*)(\log\frac{\theta_{ab}^*}{1-\theta_{ab}^*})^2).
\]
\end{theorem}

For a general $k'<k$, the same type
of limiting distribution under conditions (A1)-(A3) holds. But the proof will
involve more tedious descriptions of how various merges can occur as discussed in \cite{Wang:Bickel:2017}.

For a sparse network, we have the following result.

\begin{corollary}\label{corollary:underfit:a}
Suppose that $A\sim P_{\theta^*, z^*}$, (A1), (A2) and (A$3'$) hold, and $n \rho_n/\log n \rightarrow \infty$.  If $k=\min\{o(n\rho_n/\log (n\rho_n)), o((n/\log n)^{1/2})\}$, we have
\[
(n^{-1}L_{k,k-1}-n\mu)/\sigma(\theta^*)  \stackrel{d}{\rightarrow}  N(0,1).
\]
\end{corollary}

\subsection{$k'=k$}

For fixed $k$, we establish the Wilks theorem for a dense network.
\begin{theorem}\label{theorem:fit:a}
Suppose that $A\sim P_{\theta^*, z^*}$, conditions (A1)-(A3) hold, and $\rho_n \equiv 1$. For fixed $k$,  we have
\[
2(\max_{z\in [k]^n}\sup_{\theta\in\Theta_k}\log f(A|\theta,z)-\log f(A|\theta^*,z^*)) \stackrel{d}{\rightarrow} \chi_{\frac{k(k+1)}{2}}^2.
\]
\end{theorem}

For increasing $k$, we have the following two results which will be used to establish the consistency of the CBIC. 

\begin{corollary}\label{corollary:fit:b}
Suppose that $A\sim P_{\theta^*, z^*}$, conditions (A1)-(A3) hold, and $\rho_n \equiv 1$.  If $k=o(n/\log n)$, we have
\[
2(\max_{z\in [k]^n}\sup_{\theta\in\Theta_k}\log f(A|\theta,z)-\log f(A|\theta^*,z^*))=O_p (k^2\log k).
\]
\end{corollary}

For a sparse network, we have the following result.

\begin{corollary}\label{corollary:fit:a}
Suppose that $A\sim P_{\theta^*, z^*}$, (A1), (A2) and (A$3'$) hold, and $n \rho_n/\log n \rightarrow \infty$. If $k=o(n\rho_n/\log (n\rho_n))$, we have
\[
2(\max_{z\in [k]^n}\sup_{\theta\in\Theta_k}\log f(A|\theta,z)-\log f(A|\theta^*,z^*))=O_p (\rho_nk^2\log k ).
\]
\end{corollary}

\subsection{$k'>k$}

As discussed in \cite{Wang:Bickel:2017}, it is difficult to obtain the asymptotic distribution of $L_{k,k'}$  in the case $k'>k$.
Instead, we obtain its asymptotic order, which is a generalization of Theorem 2.10 in \cite{Wang:Bickel:2017} to  increasing $k$.

\begin{theorem}\label{theorem:overfit:a}
Suppose that $A\sim P_{\theta^*, z^*}$, conditions (A1)-(A3) hold, and $\rho_n \equiv 1$. If $k=o(n^{1/2})$, we have
\[
\begin{array}{lll}
L_{k,k'}& \leq \alpha n\log k' +\sup_{\theta\in\Theta_k}\log f(A|\theta,z^*)-\log f(A|\theta^*,z^*)\\
&=\alpha n\log k' + O_p (k^2\log k)\\
&=\alpha n\log k'(1+ o_p (1)),
\end{array}
\]
\end{theorem}
where $0<\alpha\leq1-\frac{C}{\log k'}+\frac{2\log n+\log k}{n\log k'}$.


For a sparse network, we have the following result.

\begin{corollary}\label{corollary:overfit:a}
Suppose that $A\sim P_{\theta^*, z^*}$, (A1), (A2) and (A$3'$) hold, and $n \rho_n/\log n \rightarrow \infty$. If $k=o((n/\rho_n)^{1/2})$, we have
\[
\begin{array}{lll}
L_{k,k'}& \leq \alpha n\log k' +\sup_{\theta\in\Theta_k}\log f(A|\theta,z^*)-\log f(A|\theta^*,z^*)\\
&=\alpha n\log k' + O_p (\rho_nk^2\log k)\\
&=\alpha n\log k'(1+ o_p (1)).
\end{array}
\]
\end{corollary}

\section{Consistency of the CBIC}
In this section, we establish the consistency of the CBIC in the sense that it chooses the correct $k$ with probability tending to one when $n$ goes to infinity.

To obtain the consistency of the CBIC, we  need an additional condition.

\label{section:modelselection}
(A4) (Consistency  condition) $n\mu/\log k\rightarrow-\infty$, for $k'<k$.

Note that $\mu \leq 0$ is clearly true asymptotically when $k'<k$ since it is the expectation of $\frac{1}{n^2} L_{k,k'}$. What we assume here is $\mu$ being bounded away from 0 or going to 0 at a rate slower than $\log k/n$.
\begin{theorem}\label{theorem:modelselection:a}
Suppose that $A\sim P_{\theta^*, z^*}$, (A1)-(A4) hold, and $\rho_n \equiv 1$. Let $\ell(k)$ be the penalized likelihood function for the CBIC, defined as in \eqref{eq:penalized:cbic}.  If $k=o((n/\log n)^{1/2})$,\\
for $k'<k$, we have
\[
P(\ell(k')>\ell(k))\rightarrow0;
\]
for $k'>k$, when $\lambda> (\alpha \log k')/ (\log k'-\log k)$, we have
\[
P(\ell(k')>\ell(k))\rightarrow0,
\]
where $\alpha$ is given in Theorem \ref{theorem:overfit:a}.
\end{theorem}

For sparse networks, we have the following results.
\begin{corollary}\label{corollary:modelselection:a}
Suppose that $A\sim P_{\theta^*, z^*}$, (A1), (A2), (A$3'$), (A$4$) hold, and $n \rho_n/\log n \rightarrow \infty$. Let $\ell(k)$ be the penalized likelihood function for the CBIC, defined as in \eqref{eq:penalized:cbic}. If $k=\min\{o(n\rho_n/\log (n\rho_n)), o( (n/\log n)^{1/2})\}$,\\
for $k'<k$, we have
\[
P(\ell(k')>\ell(k))\rightarrow0;
\]
for $k'>k$, when $\lambda> (\alpha \log k')/ (\log k'-\log k)$, we have
\[
P(\ell(k')>\ell(k))\rightarrow0.
\]
\end{corollary}

By Theorem \ref{theorem:modelselection:a},
the probability $P(\ell(k')>\ell(k))$ goes to zero, regardless of the value of the tuning parameter $\lambda$ in the case of $k'<k$.
When $k'>k$, it depends on the parameter $\lambda$. Then a natural question is whether  $\lambda=1$ is  a good choice. Note that it also depends on $\alpha$.
With an appropriate $\alpha$, the probability $P(\ell(k')>\ell(k))$ also goes to zero when $\lambda=1$ for fixed $k$ as demonstrated in
the following corollary.

\begin{corollary}
Suppose that $A\sim P_{\theta^*, z^*}$, (A1)-(A4) hold, and $\rho_n \equiv 1$. Let $\ell(k)$ be the penalized likelihood function for the CBIC, defined as in \eqref{eq:penalized:cbic}. If $k$ is fixed, \\
for $k'<k$, we have
\[
P(\ell(k')>\ell(k))\rightarrow0;
\]
for $k'>k$, suppose $\alpha<1-\frac{\log k}{\log k'}$, for $\lambda=1$, we have
\[
P(\ell(k')>\ell(k))\rightarrow0.
\]
\end{corollary}

By checking the proof of Theorem \ref{theorem:modelselection:a}, it is not difficult to see that for $k'>k$, $P(\tilde{\ell}(k')>\tilde{\ell}(k))\rightarrow1.$ This implies that the BIC tends to overestimate the number of communities $k$ for stochastic block models.

\section{Extension to a degree-corrected SBM}
\label{section:extension}
Real-world networks often include a number of high degree ``hub'' nodes that have many connections \citep{Barabasi:Bonabau:2003}.
To incorporate the degree heterogeneity within communities, the degree corrected stochastic block model (DCSBM)
was proposed by \cite{Karrer:Newman:2011}. Specifically,
this model assumes that $P(A_{ij}=1\mid z(i),z(j))=\omega_i\omega_j\theta_{z(i)z(j)}$, where
$\mathbf{\omega}=(\omega_i)_{1\leq i \leq n}$ are a set of node degree parameters measuring
the degree variation within blocks. For identifiability of the model, the constraint $\sum_i\omega_i\mathbf{1}\{z(i)=a\}=n_a$ can be imposed for each community $1\leq a\leq k$.

As in \cite{Karrer:Newman:2011}, we replace the Bernoulli random variables $A_{ij}$ by the Poisson random variable.
As discussed in \cite{Zhao:Levina:Zhu:2012}, there is no practical difference with respect to performance.
The reason is that the Bernoulli distribution with a small mean is well approximated by the Poisson distribution. An advantage of using Poisson distributions is that it will greatly simplify the calculations.
Another advantage is that it will allow networks to contain both multi-edges and self-edges.

For any fixed $(\theta,\omega, z)$, the log-likelihood
of observing the adjacency matrix $A$ under the degree corrected stochastic block model is
\[
\log f(A|\theta,\omega,z)=\sum_{1\leq i\leq n}d_i\log\omega_i+\sum_{1\leq a\leq b\leq k}(m_{ab}\log\theta_{ab}-n_{ab}\theta_{ab}),
\]
where $d_i=\sum_{1\leq j\leq n}A_{ij}$.

We first consider the case $\omega$ is known, which was also assumed by \cite{Lei:2016} and \cite{Gao:Ma:Zhang:Zhou:2016} in their theoretical analyses. With similar arguments, one can show that the previous Theorems \ref{theorem:underfit:a} and \ref{theorem:overfit:a}  still hold in the DCSBM.
Although Theorem \ref{theorem:fit:a} does not hold in the DCSBM, we have the following result.

\begin{theorem}\label{theorem:extension}
Suppose that $A\sim P_{\theta^*, z^*}$, (A1)-(A3) hold, and $\rho_n \equiv 1$. If $k=o(n/\log n)$, we have
\[\max_{z\in [k]^n}\sup_{\theta\in\Theta_k}\log f(A|\theta,\omega,z)-\log f(A|\theta^*,\omega,z^*)=O_p(k^2\log k).\]
\end{theorem}
Therefore, Theorem \ref{theorem:modelselection:a} still holds in the DCSBM.

If $\omega_i$'s are unknown, we use a plug-in method. That is, we need to estimate $\omega_i$'s.
After imposing the identifiability constrain on $\mathbf{\omega}$, the MLE of the parameter $\omega_i$
is given by $\hat{\omega}_i=n_ad_i/\sum_{j:z_j=z_i}d_j$.
Simulation studies indicate that the CBIC can estimate $k$ with high accuracy for the DCSBM.

\section{Experiments}
\label{section:experiments}

\subsection{Algorithm}
Since there are $k^n$ possible assignments for the communities,
it is intractable to directly optimize the log-likelihood of the SBM. Since the primary goal of our article is to study the penalty function, we use a computationally feasible algorithm -- spectral clustering to estimate the community labels for a given $k$.

The algorithm finds the eigenvectors $u_1,\ldots,u_k$ associated with the $k$ eigenvalues of the Laplacian matrix that are largest in magnitude, forming an $n\times k$ matrix $U=(u_1,\ldots,u_k)$, and then applies the $k$-means algorithm to the rows of $U$. For details, see \cite{Rohe:Chatterjee:Yu:2011}.
They established the consistency of spectral clustering in the stochastic block model under proper conditions on the density of the network and the eigen-structure of the Laplacian matrix.

For the DCSBM, we apply a variant of spectral clustering, called spectral clustering on ratios-of-eigenvectors (SCORE) proposed by \cite{Jin:2015}.
Instead of using the Laplacian matrix, the SCORE collects the eigenvectors $v_1,\ldots,v_k$ associated with the $k$ eigenvalues of $A$ that are largest in magnitude, and then forms the $n\times k$ matrix $V=(\mathbf{1},v_2/v_1,\ldots,v_k/v_1)$. The SCORE then applies the $k$-means algorithm to the rows of $V$. The corresponding consistency results for the DCSBM were also established by \cite{Jin:2015}.

We restrict our attention to candidate values for the true number of communities in the range $k'\in\{1,\ldots,18\}$, both in simulations and the real data analysis.

\subsection{Simulations}

\begin{figure}
\centering
\includegraphics[width=2.8in,height=2.2in]{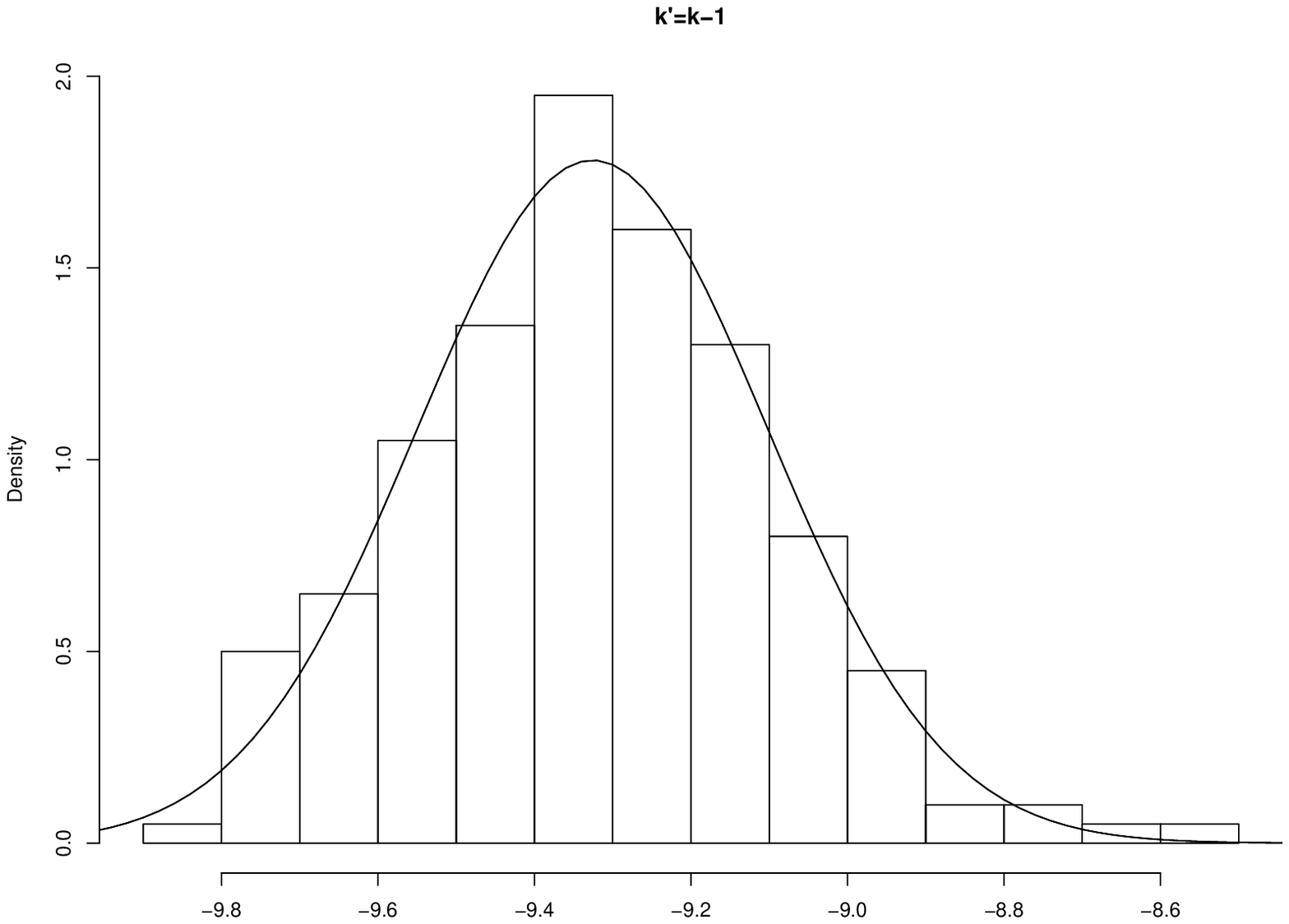}
\caption{Emipirical distribution of $n^{-1}L_{k,k-1}$.
The solid curve is normal density with mean $n\mu$ and $\sigma(\theta^*)$ as given in Theorem \ref{theorem:underfit:a}. }
\label{figure-a}
\end{figure}

\begin{figure}
\centering
\includegraphics[width=2.8in,height=2.2in]{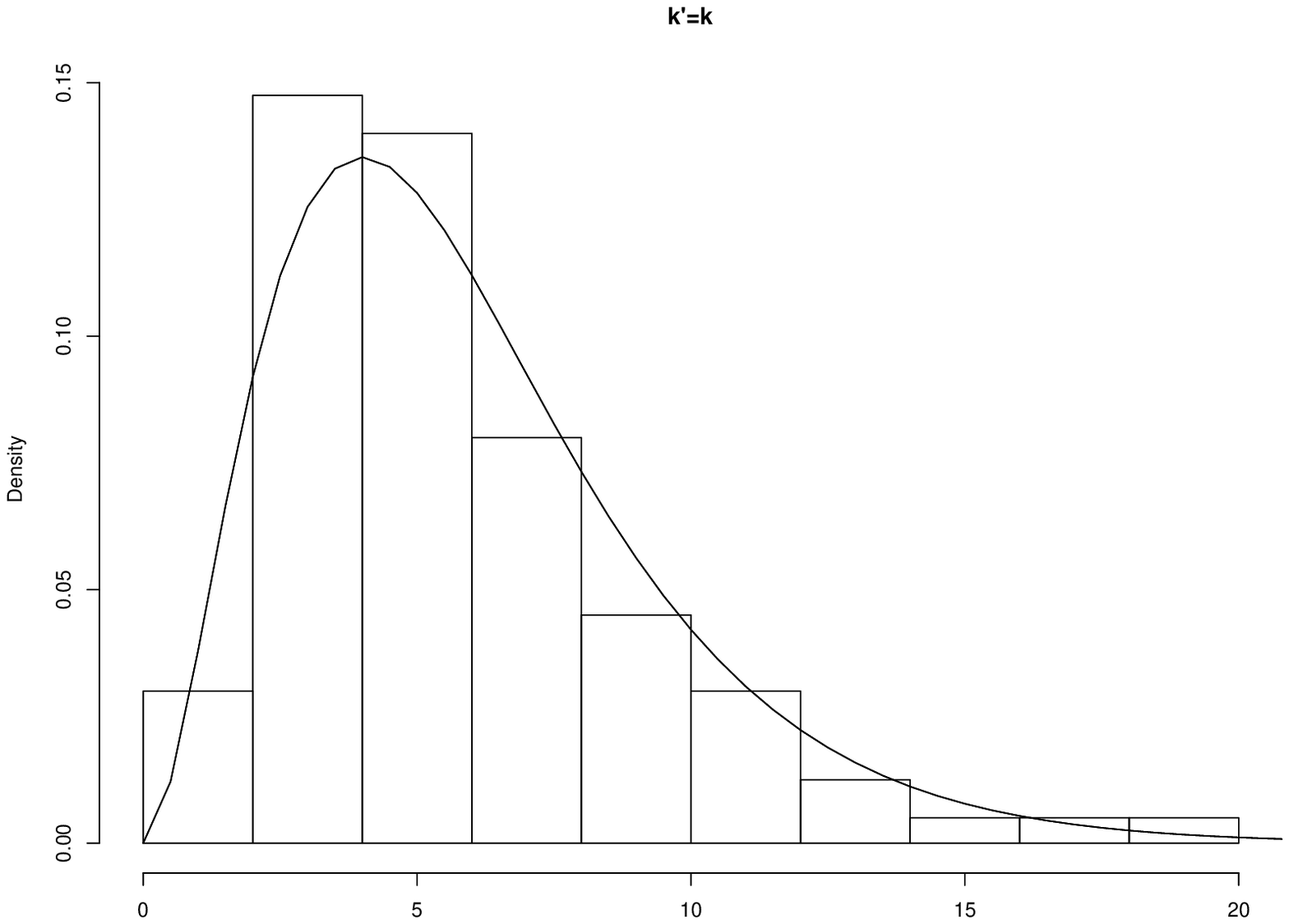}
\caption{Emipirical distribution of $2L_{k,k}$.  The solid curve is chi-square density with degree $\frac{k(k+1)}{2}=6$.  }
\label{figure-b}
\end{figure}

\begin{figure}
\centering
\includegraphics[width=2.8in,height=2.2in]{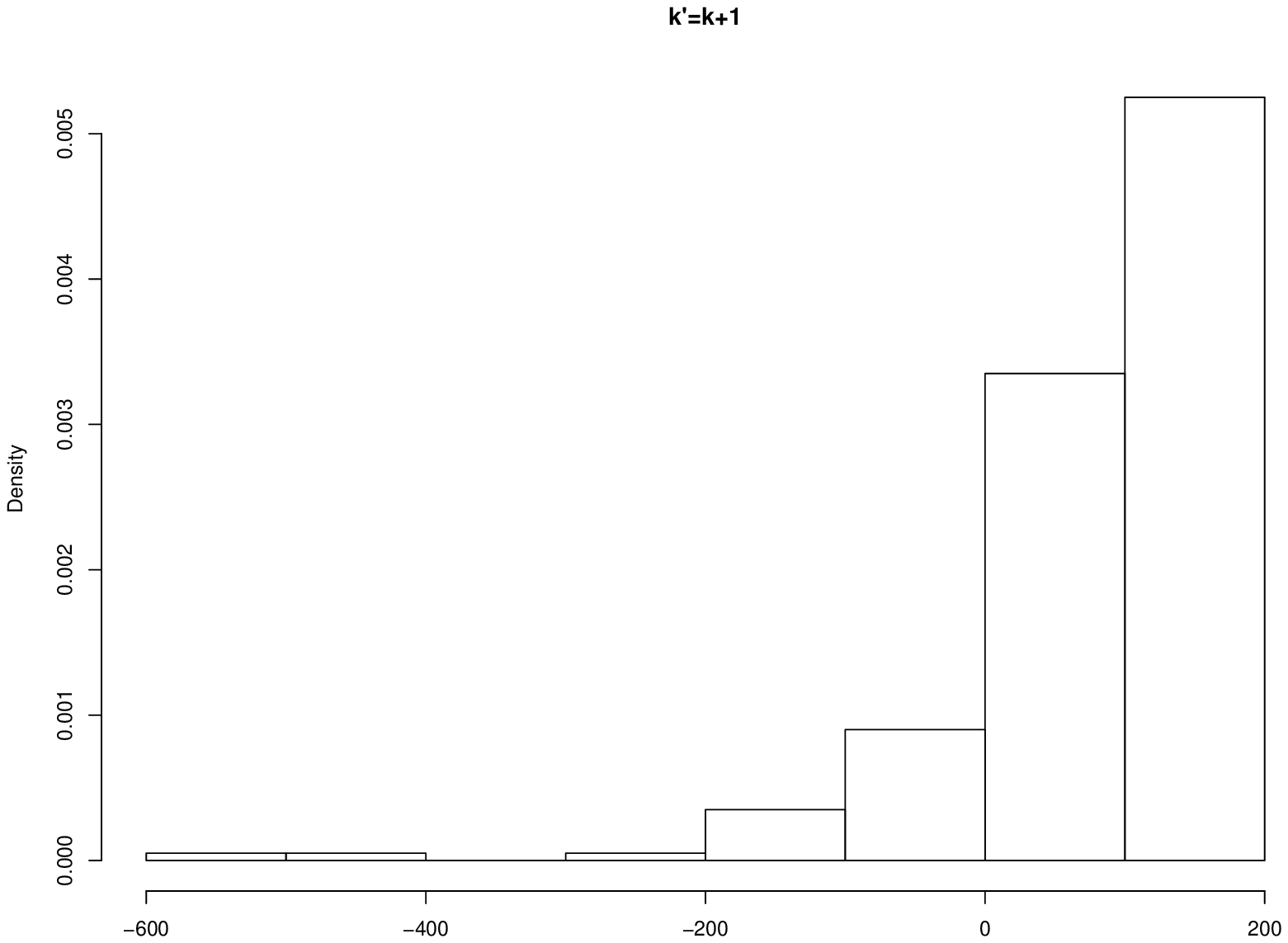}
\caption{Emipirical distribution of $2L_{k,k+1}$. }
\label{figure-c}
\end{figure}

\emph{Simulation 1.} In the SBM setting, we first compare the empirical distribution of the log-likelihood ratio with the asymptotic results in Theorems \ref{theorem:underfit:a},  \ref{theorem:fit:a} and \ref{theorem:overfit:a}.
We set the network size as $n=500$ and the probability matrix $\theta^*_{ab}=0.03(1+5\times\mathbf{1}(a=b))$. We set $k=3$ with $\pi_1=\pi_2=\pi_3=1/3$. Each simulation in this section is repeated 200 times.
The plot for $n^{-1}L_{k,k-1}$ is shown in Figure \ref{figure-a}.
The empirical distribution is well approximated by the normal distribution in the case of underfitting.
Figure \ref{figure-b} plots the empirical distribution of $2L_{k,k}$ in the case of $k'=k$.
The distribution also matches the chi-square distribution well.
Figure \ref{figure-c} plots the empirical distribution of $L_{k,k+1}$.

\begin{figure}
\centering
\includegraphics[width=3in,height=2.4in]{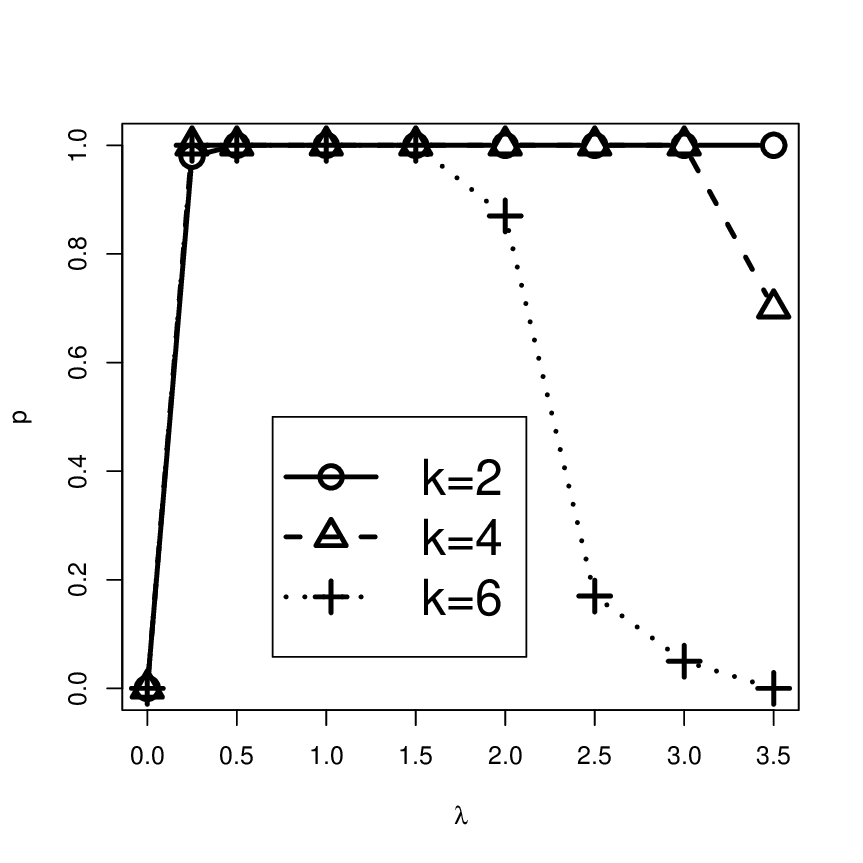}
\caption{Success rate vs $\lambda$. }
\label{figure-lambda}
\end{figure}

\emph{Simulation 2.}
In the SBM setting, we investigate how the accuracy of the CBIC changes as the tuning parameter $\lambda$ varies.
We let $\lambda$ increase from $0$ to $3.5$.
The probability matrix is the same as in Simulation 1.
We set each block size according to the sequence $(60,90,120,150,60,90,120,150)$.
That is, if $k=1$, we set the network size $n$ to be $60$; if $k=2$, we set two respective block sizes
to be $60$ and $90$; and so forth.
This setting is the same as in \cite{Saldana:Yu:Feng:2017}.
As can be seen in Figure \ref{figure-lambda}, the rate of the successful recovery of the number of communities is very low when $\lambda$ is close to zero.
When $\lambda$ is between $0.5$ and $1.5$, the success rate is almost $100\%$;
When $\lambda$ becomes larger, the success rate decreases in which the change point depends on $k$. It can be seen from Figure \ref{figure-lambda} that $\lambda=1$ is a safe tuning parameter.

\begin{table}[h!]
	\caption{Comparison of model selection methods for SBM: $r=5$}
	\label{tab:1a}
	\centering
	\scalebox{0.7}{
		\begin{tabular}{ccc|cc|cc|cc|cc|cc|cc}
			\\
			\hline
			& \multicolumn{ 2}{c|}{CBIC ($\lambda=1/4$)}	& \multicolumn{ 2}{c|}{CBIC ($\lambda=1/2$)}& \multicolumn{ 2}{c|}{CBIC ($\lambda=1$)}&          \multicolumn{ 2}{c|}{BIC }&         \multicolumn{ 2}{c|}{Lei (2016)}&         \multicolumn{ 2}{c|}{ICL} &  \multicolumn{ 2}{c}{PLH} \\
			\hline
			&       Prob &   Mean   &   Prob &   Mean   &Prob& Mean &Prob& Mean &Prob& Mean &Prob& Mean &Prob& Mean \\
			\hline
			$k=2$  &  0.92 & 2.10   &1.00  & 2.00             &1.00  & 2.00  & 0.24 & 3.15 &1.00     & 2.00    &1.00  & 2.00 &0.00 & 10.48\\
			$k=3$ & 1.00 & 3.01    &1.00  & 3.00             &1.00  & 3.00  & 0.63 & 3.57 &0.99  & 3.03 &1.00  & 3.00 & 0.02 & 9.02 \\
			$k=4$ & 1.00 & 4.00    &1.00  & 4.00              &1.00  & 4.00  & 0.78 & 4.34 &0.97  & 4.02 &1.00  & 4.00 &0.15 & 7.17 \\
			$k=5$ & 1.00 & 5.00    &1.00  & 5.00              &1.00  & 5.00  & 0.94 & 5.08 &0.86  & 4.86 &1.00  & 5.00 &0.73 &  5.70  \\
			$k=6$ &  1.00 & 6.00    &1.00  & 6.00              &1.00  & 6.00  &  0.97 & 6.03   &0.65  & 5.65 &1.00  & 6.00 &1.00 & 6.00 \\
			$k=7$ &  1.00 & 7.00    &1.00  & 7.00              &1.00  & 7.00  & 1.00 & 7.00   &0.21  & 6.21 &1.00  & 7.00 &1.00 & 7.00  \\
			$k=8$ & 1.00 & 8.00    &1.00  & 8.00              &1.00  & 8.00  & 1.00 & 8.00   &0.16  & 7.16 &1.00  & 8.00 &1.00 & 8.00 \\
			
			\hline
			
		\end{tabular}
	}
\end{table}

\begin{table}[h!]
	\caption{Comparison of model selection methods for SBM: $r=4$}
	\label{tab:1b}
	\centering
	\scalebox{0.7}{
		\begin{tabular}{ccc|cc|cc|cc|cc|cc|cc}
			\\
			\hline
			& \multicolumn{ 2}{c|}{CBIC ($\lambda=1/4$)}	& \multicolumn{ 2}{c|}{CBIC ($\lambda=1/2$)}& \multicolumn{ 2}{c|}{CBIC ($\lambda=1$)}&          \multicolumn{ 2}{c|}{BIC }&         \multicolumn{ 2}{c|}{Lei (2016)}&         \multicolumn{ 2}{c|}{ICL}  & \multicolumn{ 2}{c}{PLH} \\
			\hline
			&     Prob &   Mean   &    Prob &   Mean   &  Prob &   Mean   &Prob& Mean &Prob& Mean &Prob& Mean &Prob& Mean \\
			\hline
			$k=2$ & 0.94 & 2.08     & 1.00 & 2.00 & 1.00 & 2.00     & 0.37 & 2.92 &1.00  & 2.00 &1.00  & 2.00 & 0.00 & 10.35\\
			$k=3$ & 0.98 & 3.02    &1.00  & 3.00  & 1.00 & 3.00     & 0.73 & 3.36 &0.95  & 3.04 &1.00  & 3.00 &0.03 & 8.68 \\
			$k=4$ & 1.00 & 4.00    &1.00  & 4.00  & 1.00 & 4.00     & 0.84 & 4.26 &0.62  & 4.31 &1.00  & 4.00 &0.13 & 7.06\\
			$k=5$ & 1.00 & 5.00    &1.00  & 5.00  & 0.91 & 4.91     & 0.98 & 5.04 & 0.08 & 4.50 &0.98 & 4.99 & 0.40 & 6.71 \\
			$k=6$ & 1.00 & 6.00    &1.00  & 6.00  &  0.90 & 5.91     & 0.99 & 6.01  & 0.02  & 5.14 &0.99 & 5.99 & 0.82 & 6.38  \\
			$k=7$ & 1.00 & 7.00    &1.00   & 7.00  & 0.85 & 6.85       & 1.00 & 7.00 &0.04  & 6.00 &0.96 & 6.96 &1.00 & 7.00  \\
			$k=8$ &  0.99 & 8.00    &0.97   & 7.97  & 0.70 & 7.70       & 1.00 & 8.01 &0.05  & 6.40 & 0.86 & 7.86  &0.98 & 8.00 \\
			\hline
			
		\end{tabular}
	}
\end{table}

\begin{table}[h!]
	\caption{Comparison of model selection methods for SBM: $r=3$}
	\label{tab:1c}
	\centering
	\scalebox{0.7}{
		\begin{tabular}{ccc|cc|cc|cc|cc|cc|cc}
			\\
			\hline
			& \multicolumn{ 2}{c|}{CBIC ($\lambda=1/4$)} & \multicolumn{ 2}{c|}{CBIC ($\lambda=1/2$)}& \multicolumn{ 2}{c|}{CBIC ($\lambda=1$)}&          \multicolumn{ 2}{c|}{BIC }&         \multicolumn{ 2}{c|}{Lei (2016)}&         \multicolumn{ 2}{c|}{ICL}  & \multicolumn{ 2}{c}{PLH}\\
			\hline
			&        Prob &   Mean   &             Prob &   Mean   &           Prob &   Mean   & Prob& Mean &Prob& Mean &Prob& Mean &Prob& Mean\\
			\hline
			$k=2$     & 0.91 & 2.11              &1.00  & 2.00      & 1.00 & 2.00          & 0.33 & 3.07 &0.99  & 2.01& 1.00 & 2.00 & 0.00 & 10.03 \\
			$k=3$      & 1.00 & 3.00              &1.00  & 3.00       & 0.99 & 2.99         & 0.80  & 3.31 &0.47 & 3.00 & 1.00 & 3.00 & 0.02 &  9.35 \\
			$k=4$    & 0.99 & 4.01               &1.00   & 4.01 &  0.96 & 3.96      & 0.92 & 4.14 &0.15  & 3.33 &0.98 & 3.99 & 0.42 & 5.87\\
			$k=5$   & 0.90 & 4.92                &0.73  & 4.74  & 0.30 & 4.24      & 0.88 & 5.04 &0.13  & 3.67 &0.44 & 4.41 & 0.54 & 6.25\\
			$k=6$   & 0.55 & 5.64                & 0.45  & 5.47  & 0.14 & 4.93     & 0.57 & 5.85  &0.07  & 4.17 &0.20 & 4.95 & 0.49 & 6.65\\
			$k=7$ & 0.28 & 6.28                   &0.19 &  6.09  & 0.06 & 5.66      & 0.32 & 6.53 &0.00  & 4.89 &0.06 & 5.58 & 0.31 & 7.03\\
			$k=8$  & 0.12 & 6.86                 &0.05   &6.67  & 0.01 & 6.36      & 0.20 & 7.25 &0.00  & 5.50 &0.02 & 6.17 & 0.20 & 7.49 \\
			
			\hline
			
		\end{tabular}
	}
\end{table}

\emph{Simulation 3.} In the SBM setting, we compare the CBIC with the BIC proposed by  \cite{Saldana:Yu:Feng:2017}, the bootstrap corrected sequential test proposed by \cite{Lei:2016}, the ICL proposed by \cite{Daudin:Picard:Robin:2008} and the penalized likelihood method (PLH) proposed by \cite{Wang:Bickel:2017}. For the bootstrap corrected sequential test, we select threshold $t_n$ corresponding to the nominal type I error bound $10^{-4}$. The network size $n$ is the same as in Simulation 2 and the probability matrix is $\theta^*_{ab}=0.03(1+r\times\mathbf{1}(a=b))$. Note that in the SBM setting, the method of \cite{Lei:2016} is better than the network cross-validation of \cite{Chen:Lei:2018} (NCV) according to our simulations. Thus, in the SBM setting, the CBIC is not compared with the NCV of \cite{Chen:Lei:2018}.

The numerical results are shown in Tables  \ref{tab:1a}, \ref{tab:1b} and \ref{tab:1c}. From these tables, we can see that the CBIC shows a significant improvement over the BIC and the bootstrap corrected sequential test.
It can be seen from Table \ref{tab:1a} that, for $r=5$, the CBIC ($\lambda=1$) recovers the number of communities $k$ perfectly while the success rates for the BIC and the bootstrap corrected sequential test are low for $k\le 4$ and $k\ge 5$, respectively. It can also be seen from Table \ref{tab:1c} that, for $r=3$, the CBIC ($\lambda=1$) recovers the number of communities $k$ quite well for $k\le 4$. When the number of communities $k$ is large (e.g., $k\ge5$), for $r=3$, the BIC outperforms the CBIC. For this case, the performance of the CBIC can be improved by using a smaller $\lambda$.

Additionally, the CBIC with $\lambda=1/2$ consistently outperforms the ICL in all scenarios as shown in these tables. The CBIC with $\lambda=1/4$ performs even better for large $k$; for small $k$ ($k=2,3$), it performs slightly worse than the ICL. These results also suggest that when the number of communities is relatively small,
$\lambda=1$ is a good choice. On the other hand, when the number of communities is relatively large, $\lambda<1$, i.e., a lighter penalty is a better choice. If we have some prior knowledge about the number of communities, this observation  provides some guidance on the selection of $\lambda$.


We also compare the results from the PLH proposed by \cite{Wang:Bickel:2017} in Tables  \ref{tab:1a}, \ref{tab:1b} and \ref{tab:1c}. We found out that a typical tuning parameter selected by the procedure recommended in \cite{Wang:Bickel:2017} is usually small. Therefore, the PLH usually overestimates the number of communities $k$.

\begin{table}[h!]
\caption{Comparison of model selection methods for DCSBM: $r=5$  \label{tab:2a}}
\centering
\scalebox{1}{
\begin{tabular}{ccc|cc|cc|ccccccccccc}
\\
\hline
          & \multicolumn{ 2}{c|}{CBIC ($\lambda=1$)}&         \multicolumn{ 2}{c|}{BIC }&         \multicolumn{ 2}{c|}{NCV}&         \multicolumn{ 2}{c}{ICL}\\
\hline
&                       Prob &   Mean   &Prob& Mean &Prob& Mean &Prob& Mean\\
\hline
$k=2$                    &0.98  & 2.02  &0.11  & 3.37  & 0.94& 2.14&0.98  & 2.02 \\
$k=3$                    &0.99  & 3.01 &0.16  & 4.30  & 0.94& 3.06 &0.99  & 3.01\\
$k=4$                    &0.99  & 4.01 &0.37  & 4.91  & 0.89& 4.14 &0.99  & 4.01\\
$k=5$                    &0.97  & 5.05 &0.46  & 5.71  & 0.33& 5.09 &0.96  & 5.06\\
$k=6$                    &0.97  & 6.03 &0.38  & 6.57  & 0.29& 7.41 &0.97  & 6.03\\
$k=7$                    &0.82  & 7.11  &0.54  & 7.67& 0.25& 8.50 &0.83  & 7.10\\
$k=8$                    &0.72  & 8.09  &0.50  & 8.36 & 0.15& 9.38 &0.71  & 8.13\\

\hline

\end{tabular}
}
\end{table}

\begin{table}[h!]
\caption{Comparison of model selection methods for DCSBM: $r=4$  \label{tab:2b}}
\centering
\scalebox{1}{
\begin{tabular}{ccc|cc|cc|ccccccccccc}
\\
\hline
          & \multicolumn{ 2}{c|}{CBIC ($\lambda=1$)}&         \multicolumn{ 2}{c|}{BIC }&         \multicolumn{ 2}{c|}{NCV}&         \multicolumn{ 2}{c}{ICL}\\
\hline
&                       Prob &   Mean   &Prob& Mean &Prob& Mean &Prob& Mean\\
\hline
$k=2$                    &1  & 2 &0.16  & 3.47  & 0.80& 2.52&1  & 2 \\
$k=3$                    &0.98  & 3.02  &0.24  & 4.14 & 0.65& 3.60 &0.97  & 3.04\\
$k=4$                    &0.99  & 4.02 &0.51  & 4.70  & 0.14& 4.12&0.98  & 4.03 \\
$k=5$                    &0.93  & 5.04 &0.60  & 5.50  & 0.17& 5.42 &0.93  & 5.04\\
$k=6$                    &0.84  & 5.90 &0.63  & 6.46  & 0.16& 6.41 &0.84  & 5.90\\
$k=7$                    &0.17  & 6.43 &0.21  & 7.11  & 0.23& 8.00 &0.17  & 6.45\\
$k=8$                    &0.17  & 7.35 &0.21  & 8.20  & 0& 8.76 &0.17  & 7.32\\

\hline

\end{tabular}
}
\end{table}

\begin{table}[h!]
\caption{Comparison of model selection methods for DCSBM: $r=3$  \label{tab:2c}}
\centering
\scalebox{1}{
\begin{tabular}{ccc|cc|cc|ccccccccccc}
\\
\hline
          & \multicolumn{ 2}{c|}{CBIC ($\lambda=1$)}&         \multicolumn{ 2}{c|}{BIC }&         \multicolumn{ 2}{c|}{NCV}&         \multicolumn{ 2}{c}{ICL}\\
\hline
&                       Prob &   Mean   &Prob& Mean &Prob& Mean &Prob& Mean\\
\hline
$k=2$                    &0.95  & 2.05 &0.10  & 3.57  & 0.74& 2.44&0.92  & 2.08 \\
$k=3$                    &0.92  & 3.02 &0.19  & 4.09  & 0.16& 4.20 &0.92  & 3.02\\
$k=4$                    &0.18  & 3.38 &0.30  & 4.44  & 0.15& 3.56&0.18  & 3.36 \\
$k=5$                    &0.11  & 3.96 &0.27  & 5.24  & 0.13& 3.82 &0.10  & 3.94\\
$k=6$                    &0.07  & 4.54  &0.22  & 5.46 & 0.10& 5.56 &0.06  & 4.52\\
$k=7$                    &0.07  & 5.70 &0.18  & 6.61 & 0.07& 6.43 &0.07  & 5.70\\
$k=8$                    &0  & 6.05  &0.13  & 7.07 & 0& 9.09 &0  & 6.05\\

\hline

\end{tabular}
}
\end{table}

\emph{Simulation 5.} Now we investigate the performance of the CBIC in the DCSBM setting. Since the bootstrap corrected sequential test is only designed for the SBM, we compare the CBIC with the BIC and the NCV. In choosing the parameters $\theta,\omega$ in the DCSBM, we follow the approach proposed in \cite{Zhao:Levina:Zhu:2012}. That is, $\omega_1,\dots,\omega_n$ are independently generated from a distribution with expectation 1, specifically
$$\omega_i=\left\{
        \begin{array}{ll}
          \eta_i, & \hbox{w.p. 0.8;} \\
          7/11, & \hbox{w.p. 0.1;}\\
          15/11, & \hbox{w.p. 0.1},
        \end{array}
      \right.
$$
where $\eta_i$ is uniformly distributed on the interval $[\frac{3}{5},\frac{7}{5}]$.
The edge probability and network sizes are set the same as in Simulation 3.
The numerical results are given in Tables \ref{tab:2a}, \ref{tab:2b} and \ref{tab:2c}. The comparisons are similar to those in Tables \ref{tab:1a}, \ref{tab:1b} and \ref{tab:1c}.

\subsection{Real data analysis}
\subsubsection{International trade dataset}
We study an international trade dataset collected by \cite{Westveld:Hoff:2011}.
It contains yearly international trade data among $n=58$ countries from 1981--2000.
One can refer to \cite{Westveld:Hoff:2011} for a detailed description.
This dataset was revisited by \cite{Saldana:Yu:Feng:2017} for the purpose of estimating the number of communities.
Following their paper, we only focus on data from 1995 and transform the weighted adjacency matrix to
the binary matrix using their methods.
An adjacency matrix $A$ is created by first considering a weight matrix $W$ with $W_{ij}=\mbox{Trade}_{ij}+\mbox{Trade}_{ji}$,
where $\mbox{Trade}_{ij}$ denotes the value of exports from country $i$ to country $j$.
Define $A_{ij}=1$ if $W_{ij}\geq W_\alpha$, and $A_{ij}=0$ otherwise. Here $W_\alpha$ denotes the $\alpha$-th quantile of $\{W_{ij}\}_{1\leq i<j\leq n}$.
We set $\alpha=50$ as in \cite{Saldana:Yu:Feng:2017}.
At $\lambda=1$, the CBIC for the SBM estimates $\hat{k}=5$, while the BIC and the NCV estimate $\hat{k}=10$ and $\hat{k}=3$, respectively.  The CBIC for the DCSBM estimates $\hat{k}=3$, while both the BIC and the NCV estimate $\hat{k}=1$. As discussed in \cite{Saldana:Yu:Feng:2017}, it seems reasonable to select 3 communities, corresponding to countries with highest GDPs, industrialized European and Asian countries with medium-level GDPs, and developing countries in South America with the smallest GDPs.

\subsubsection{Political blog dataset}
We study the political blog network \citep{Adamic:Glance:2005}, collected around 2004. This network
consists of blogs about US politics, with edges representing web links.
The nodes are labeled as ``conservative'' and ``liberal'' by the authors of \cite{Adamic:Glance:2005}. So it is reasonable to assume that this network contains these two communities. We only consider its largest
connected component of this network which consists of $1222$ nodes with community sizes $586$ and $636$ as is commonly done in the literature.
It is widely believed that the DCSBM is a better fit for this network than the SBM.
At $\lambda=1$, the CBIC for the DCSBM estimates $\hat{k}=2$, while the PLH and the NCV estimate $\hat{k}=1$ and $\hat{k}=2$, respectively. We can see that both the CBIC and the NCV give a reasonable estimate for the number of communities.

\section{Discussion}
\label{section:discussion}
In this paper, under both the SBM and the DCSBM,
we have proposed a ``corrected Bayesian information criterion" that leads to a consistent estimator for the number of communities.
The criterion improves those used in \cite{Wang:Bickel:2017} and \cite{Saldana:Yu:Feng:2017} which tend
to underestimate and overestimate the number of communities, respectively.
The simulation results
indicate that the criterion has a good performance for estimating the number of communities for finite sample sizes.

Some extensions of the research in this article are possible. For instance, it is interesting to study whether the CBIC is still consistent for correlated binary data.
For this case,  we plan to study the composite likelihood studied in \cite{Saldana:Yu:Feng:2017}.

 Furthermore, we have noticed that $\lambda=1$ is not always the best choice.  When the number of communities $k$ is large (e.g., $k\ge5$), for both medium and small $r$ (e.g., $1<r\leq3$), $\lambda=1$ tends to underestimate the number of communities. As a result, $0\leq\lambda<1$ may be a better choice. For this case, we may use other methods to choose the tuning parameter $\lambda$, which will be explored for future work.

Finally, the theoretical studies in this article focus on the maximum likelihood  estimator of the SBM. It is well-known that achieving the exact maximum is an NP-hard problem \citep{Amini:etal:2013}. Many computationally efficient methods, such as the methods proposed by \cite{Amini:etal:2013} or \cite{Rohe:Chatterjee:Yu:2011} can achieve weakly consistency. Theoretically, whether the error introduced by the approximation affects the asymptotic consistency is an open problem. Although a general theory for these estimators may be difficult, for future work, we plan to  study the model selection consistency for specific algorithms.

\section{Appendix}
\label{section:appendix}

For simplicity, we first consider the case $\rho_n \equiv 1$. By using the techniques developed in \cite{Wang:Bickel:2017}, the case for $\rho_n\rightarrow 0$ at the rate $n \rho_n/\log n \rightarrow \infty$ can be shown in a similar way.

We quote some notations from \cite{Wang:Bickel:2017}. Define
\[
F(M,t)=\sum_{1\leq a\leq b\leq k'}t_{ab}\gamma(\frac{M_{ab}}{t_{ab}}),
\]

where $\gamma(x)=x\log x+(1-x)\log (1-x)$.

Define
\[
G(R(z),\theta^*)=\sum_{1\leq a\leq b\leq k'}[R\mathbf{1}\mathbf{1}^TR^T(z)]_{ab}\gamma(\frac{[R\theta^*R^T(z)]_{ab}}{[R\mathbf{1}\mathbf{1}^TR^T(z)]_{ab}}),
\]
where  $R(z)$ is the $k'\times k$ confusion matrix whose $(a,b)$-entry is
\[
R_{ab}(z,z^*)=\frac{1}{n}\sum_{i=1}^n\sum_{j\neq i}\mathbf{1}\{z_i=a,z_j^*=b\}.
\]
$G(R(z),\theta^*)$ can be viewed as a ``population version'' of the profile likelihood. That is, roughly speaking, $G(R(z),\theta^*)$ is the expected value of $F(m(z)/n^2,n(z)/n^2)$ under $\theta^*$.
The following  Lemmas \ref{lemma:underfit:a} and \ref{lemma:underfit:b} are essentially from \cite{Wang:Bickel:2017},
which bound the variations in $A$ and will be used in the proofs of Theorems \ref{theorem:underfit:a}-\ref{theorem:overfit:a}. For more work on this topic, we refer to \cite{Bickel:Chen:2009} and \cite{Bickel:etal:2013}.

Lemma \ref{lemma:underfit:a} shows in the case of underfitting an SBM with $(k-1)$ communities, $G(R(z),\theta^*)$ is maximized by combining two existing communities in the true model.

\begin{lemma}\label{lemma:underfit:a}
Given the true labels $z^*$, maximizing the function $G(R(z),\theta^*)$ over $R$ achieves its maximum in the label set
$$\{z\in[k-1]^n: there\,\, exists\,\, \mathscr{T}\,\, such \,\,that \,\,\mathscr{T}(z)=U_{a,b}(z^*), 1\leq a<b\leq k\},$$
where $U_{a,b}$ merges $z^*_i$ with labels $a$ and $b$.

Furthermore, suppose $z'$ gives the unique maximum (up to a permutation $\mathscr{T}$), for all $R$ such that $R\geq0$, $R^T\mathbf{1}=p \mathbf{1}$,
$$\frac{\partial G((1-\epsilon)R(z')+\epsilon R(z),\theta^*)}{\partial\epsilon}\mid_{\epsilon=0+}<-C_2<0.$$
\end{lemma}

For simplicity, $R(z)\theta^*R(z)^T$ is abbreviated to $R\theta^*R(z)^T$.
\begin{lemma}\label{lemma:underfit:b}
Suppose $z\in [k']^n$ and define $X(z)=\frac{m(z)}{n^2}-R\theta^*R(z)^T$. For $\epsilon\leq3$,
$$P(\sum_{1\leq a\leq b\leq k'}|X_{ab}(z)|\geq\epsilon)\leq2(k')^{n+2}\exp(-C_1(\theta^*)n^2\epsilon^2).$$
Let $y\in[k']^n$ be a fixed set of labels, then for $\epsilon\leq \frac{3m}{n}$\footnote{This $m$ is an integer and is not to be confused with the function $m(z)$. },
$$\begin{array}{lll}P(\max_{z:\mid x-y\mid\leq m}\parallel X(z)-X(y)\parallel_{\infty}>\epsilon)\\
\leq2{{n}\choose{m}}(k')^{m+2}\exp(-C_2(\theta^*)\frac{n^3\epsilon^2}{m}),
\end{array}$$
where $C_1(\theta^*)$ and $C_2(\theta^*)$ are constants depending only on $\theta^*$.
\end{lemma}

\subsection{Proofs for Theorem \ref{theorem:underfit:a}}

In order to prove Theorem \ref{theorem:underfit:a}, we need one lemma below.

\begin{lemma}
\label{lemma:underfit:c}
Suppose that $A\sim P_{\theta^*, z^*}$. If $k=o(n/\log n)$, with probability tending to 1, we have
$$\max_{z\in [k-1]^n}\sup_{\theta\in\Theta_{k-1}}\log f(A|\theta,z)=\sup_{\theta\in\Theta_{k-1}}\log f(A|\theta,z')$$
\end{lemma}

\begin{proof}
The arguments are similar to those for Lemma 2.3 in \cite{Wang:Bickel:2017}.
By Lemma \ref{lemma:underfit:a}, without loss of generality assume the maximum of $G(R(z),\theta^*)$ is achieved at $z'=U_{k-1,k}(z^*)$. Denote $\theta'=U_{k-1,k}(\theta^*,p)$. Similar to \cite{Bickel:etal:2013}, we prove this by considering $z$ far from $z'$ and close to $z'$ (up to permutation $\mathscr{T}$). Define
$$I_{\delta_n}^-=\{z\in[k-1]^n:G(R(z),\theta^*)-G(R(z'),\theta^*)<-\delta_n\},$$
for $\delta_n\rightarrow 0$ slowly enough.

By Lemma \ref{lemma:underfit:b},
\[
\begin{array}{lll} &  & \mid F(m(z)/n^2,n(z)/n^2)-G(R(z),\theta^*)\mid\\
& \leq & C\sum_{1\leq a\leq b\leq k-1}\mid m_{ab}(z)/n^2-(R\theta^*R^T(z))_{ab}\mid\\
& =  & O_p((\log n/n)^{1/2})
\end{array}
\]
since $\gamma(\cdot)$ is Lipschitz on any interval bounded away from $0$ and $1$.

For $z\in I_{\delta_n}^-$, we have
\[
\begin{array}{lll} &  & \max_{z\in I_{\delta_n}^-}\sup_{\theta\in\Theta_{k-1}}\log f(A|\theta,z)\\
& \leq & \log(\sum_{z\in I_{\delta_n}^-}\sup_{\theta\in\Theta_{k-1}}f(A|\theta,z))\\
& =  & \log(\sum_{z\in I_{\delta_n}^-}\sup_{\theta\in\Theta_{k-1}}e^{\log f(A|\theta,z)})\\
& \leq & \log (\sup_{\theta\in\Theta_{k-1}}f(A|\theta,z')(k-1)^ne^{O_p(n^2(\log n/n)^{1/2})-n^2\delta_n})\\
& =  & \log (\sup_{\theta\in\Theta_{k-1}}f(A|\theta,z'))+\log ((k-1)^ne^{O_p(n^2(\log n/n)^{1/2})-n^2\delta_n})\\
& < & \sup_{\theta\in\Theta_{k-1}}\log f(A|\theta,z'),
\end{array}
\]
choosing $\delta_n \rightarrow 0$ slowly enough such that $\delta_n/(\log n/n)^{1/2}\rightarrow \infty$.

For $z\notin I_{\delta_n}^-$, $|G(R(z),\theta^*)-G(R(z'),\theta^*)|\rightarrow0$. Let $\bar{z}=\min_\mathscr{T}\mid\mathscr{T}(z)-z'\mid$. Since the maximum is unique up to $\mathscr{T}$, $\parallel R(\bar{z})-R(z')\parallel_{\infty}\rightarrow0$.

By Lemma \ref{lemma:underfit:b},
\[
\begin{array}{lll}
&&P(\max_{z\notin \mathscr{T}(z')}\parallel X(\bar{z})-X(z')\parallel_{\infty}>\epsilon\mid \bar{z}-z'\mid/n)\\
& \leq & \sum_{m=1}^nP(\max_{z:z=\bar{z},\mid \bar{z}-z'\mid=m}\parallel X(\bar{z})-X(z')\parallel_{\infty}>\epsilon\frac{m}{n})\\
& \leq & \sum_{m=1}^n2(k-1)^{k-1}n^m(k-1)^{m+2}e^{-Cnm} \rightarrow 0.
\end{array}
\]
It follows for $\mid \bar{z}-z'\mid=m$, $z\notin I_{\delta_n}^-$,
\[
\begin{array}{lll}
\parallel\frac{m(\bar{z})}{n^2}-\frac{m(z')}{n^2}\parallel_{\infty}
& = & o_p(1)\frac{\mid \bar{z}-z'\mid}{n}+\parallel R\theta^*R^T(\bar{z})-R\theta^*R^T(z')\parallel_{\infty}\\
& \geq & \frac{m}{n}(C+o_p(1)).
\end{array}
\]

Observe that $\parallel\frac{m(z')}{n^2}-R\theta^*R^T(z')\parallel_{\infty}=o_p(1)$. By Lemma \ref{lemma:underfit:b}, $\parallel\frac{n(z')}{n^2}-R\mathbf{1}\mathbf{1}^TR^T(z')\parallel_{\infty}=o_p(1)$.
Note that $F(\cdot,\cdot)$ has continuous derivative in the neighborhood of $(\frac{m(z')}{n^2},\frac{n(z')}{n^2})$. By Lemma \ref{lemma:underfit:a},
$$\frac{\partial F((1-\epsilon)\frac{m(z')}{n^2}+\epsilon M,(1-\epsilon)\frac{n(z')}{n^2}+\epsilon t)}{\partial \epsilon}\mid_{\epsilon=0^+}<-C<0$$
for $(M,t)$ in the neighborhood of $(\frac{m(z')}{n^2},\frac{n(z')}{n^2})$. Hence,
\[
F(\frac{m(\bar{z})}{n^2},\frac{n(\bar{z})}{n^2})-F(\frac{m(z')}{n^2},\frac{n(z')}{n^2})\leq -C\frac{m}{n}.
\]
Since
\[
\begin{array}{lll}
&&\sup_{\theta\in\Theta_{k-1}}\log f(A|\theta,z)-\sup_{\theta\in\Theta_{k-1}}\log f(A|\theta,z')\\
&\leq & n^2(F(\frac{m(\bar{z})}{n^2},\frac{n(\bar{z})}{n^2})-F(\frac{m(z')}{n^2},\frac{n(z')}{n^2}))\\
& = & -Cmn,
\end{array}
\]
we have
\[
\begin{array}{lll}
&&\max_{z\notin I_{\delta_n}^-, z\notin \mathscr{T}(z')}\sup_{\theta\in\Theta_{k-1}}\log f(A|\theta,z)\\
& \leq & \log(\sum_{z\notin I_{\delta_n}^-, z\notin \mathscr{T}(z')}\sup_{\theta\in\Theta_{k-1}}f(A|\theta,z))\\
& \leq & \log(\sum_{z\in \mathscr{T}(z')}\sup_{\theta\in\Theta_{k-1}}f(A|\theta,z)\sum_{m=1}^n(k-1)^mn^me^{-Cmn})\\
& \leq & \log(\sup_{\theta\in\Theta_{k-1}}f(A|\theta,z')\sum_{z\in \mathscr{T}(z')}\sum_{m=1}^n(k-1)^mn^me^{-Cmn})\\
& = & \sup_{\theta\in\Theta_{k-1}}\log f(A|\theta,z')+\log((k-1)^{k-1}\sum_{m=1}^n(k-1)^mn^me^{-Cmn})\\
& < & \sup_{\theta\in\Theta_{k-1}}\log f(A|\theta,z')+\log((k-1)^{k}n^2e^{-Cn})\\
& = & \sup_{\theta\in\Theta_{k-1}}\log f(A|\theta,z')+k\log(k-1)+2\log n-Cn\\
& < & \sup_{\theta\in\Theta_{k-1}}\log f(A|\theta,z').
\end{array}
\]
\end{proof}

\begin{proof}[Proof of Theorem \ref{theorem:underfit:a}]
By \citeauthor{Hoeffding:1963}'s (\citeyear{Hoeffding:1963}) inequality, we have
\[
\begin{array}{lll}
P(\max_{1\leq a\leq b\leq k}\mid \theta_{ab}^*-\hat{\theta}_{ab}\mid>t) & \leq & \sum_{1\leq a\leq b\leq k}P(\mid \theta_{ab}^*-\hat{\theta}_{ab}\mid>t)\\
& \leq & \sum_{1\leq a\leq b\leq k}e^{-2t^2n_an_b}\\
& \leq & e^{2\log k-2C_1^2n^2t^2/k^2}.
\end{array}
\]
It implies that
\[
\max_{1\leq a\leq b\leq k}\mid \theta_{ab}^*-\hat{\theta}_{ab}\mid=O_p(\frac{k\sqrt{\log k}}{n}).
\]

Note that $\sup_{\theta\in\Theta_{k-1}}\log f(A|\theta,z')$ is uniquely maximized at
$$\hat{\theta}_{ab}=\frac{m_{ab}(z')}{n_{ab}(z')}=\frac{m_{ab}}{n_{ab}}=\theta_{ab}^*+O_p(\frac{k\sqrt{\log k}}{n})\,\,\, \mathrm{for}\,\,\, 1\leq a\leq b\leq k-2,$$
$$\hat{\theta}_{a(k-1)}'=\frac{m_{a(k-1)}(z')}{n_{a(k-1)}(z')}=\frac{m_{a(k-1)}+m_{ak}}{n_{a(k-1)}+n_{ak}}=\theta_{a(k-1)}'+O_p(\frac{k\sqrt{\log k}}{n})\,\,\, \mathrm{for}\,\,\, 1\leq a\leq k-2,$$
$$\hat{\theta}_{(k-1)(k-1)}'=\frac{m_{(k-1)(k-1)}(z')}{n_{(k-1)(k-1)}(z')}=\frac{m_{(k-1)(k-1)}+m_{(k-1)k}+m_{kk}}{n_{(k-1)(k-1)}+n_{(k-1)k}+n_{kk}}=\theta_{(k-1)(k-1)}'+O_p(\frac{k\sqrt{\log k}}{n}).$$
Thus, we have
\[
\begin{array}{lll}
&&n^{-1}(\max_{z\in [k-1]^n}\sup_{\theta\in\Theta_{k-1}}\log f(A|\theta,z)-\log f(A|\theta^*,z^*))\\
&=& n^{-1}(\sup_{\theta\in\Theta_{k-1}}\log f(A|\theta,z')-\log f(A|\theta^*,z^*))\\
&=& n^{-1}(\sum_{1\leq a\leq b\leq k-2}(m_{ab}\log\frac{\hat{\theta}_{ab}}{1-\hat{\theta}_{ab}}+n_{ab}\log(1-\hat{\theta}_{ab}))\\
&&-\sum_{1\leq a\leq b\leq k-2}(m_{ab}\log\frac{\theta_{ab}^*}{1-\theta_{ab}^*}
+n_{ab}\log(1-\theta_{ab}^*))\\
&&+\sum_{k-1\leq a\leq b\leq k}(m_{ab}(z')\log\frac{\hat{\theta}_{ab}'}{1-\hat{\theta}_{ab}'}+n_{ab}(z')\log(1-\hat{\theta}_{ab}'))\\
&&-\sum_{k-1\leq a\leq b\leq k}(m_{ab}\log\frac{\theta_{ab}^*}{1-\theta_{ab}^*}
+n_{ab}\log(1-\theta_{ab}^*))).
\end{array}
\]
Let
\[
K=\sum_{1\leq a\leq b\leq k-2}(m_{ab}\log\frac{\hat{\theta}_{ab}}{1-\hat{\theta}_{ab}}+n_{ab}\log(1-\hat{\theta}_{ab}))
-\sum_{1\leq a\leq b\leq k-2}(m_{ab}\log\frac{\theta_{ab}^*}{1-\theta_{ab}^*}+n_{ab}\log(1-\theta_{ab}^*)),
\]
\[
K_1=\frac{1}{2}\sum_{1\leq a\leq b\leq k-2}\frac{n_{ab}(\hat{\theta}_{ab}-\theta_{ab}^*)^2}{\theta_{ab}^*(1-\theta_{ab}^*)}.
\]
By the proof of Theorem \ref{theorem:fit:a}, we have
\[
K=K_1+O_p(\frac{k^3\log^{3/2}k}{n})=O_p(k^2\log k)+O_p(\frac{k^3\log^{3/2}k}{n}).
\]

Thus, we have
\[
\begin{array}{lll}
&& (n^{-1}(\max_{z\in [k-1]^n}\sup_{\theta\in\Theta_{k-1}}\log f(A|\theta,z)-\log f(A|\theta^*,z^*))-n\mu)/\sigma(\theta^*)\\
&=& (n^{-1}(\sum_{k-1\leq a\leq b\leq k}(m_{ab}(z')\log\frac{\theta_{ab}'}{1-\theta_{ab}'}+n_{ab}(z')\log(1-\theta_{ab}'))\\
&&-\sum_{k-1\leq a\leq b\leq k}(m_{ab}\log\frac{\theta_{ab}^*}{1-\theta_{ab}^*}
+n_{ab}\log(1-\theta_{ab}^*)))-n\mu+O_p(\frac{K}{n}))/\sigma(\theta^*)\\
&=& (n^{-1}(\sum_{k-1\leq a\leq b\leq k}(m_{ab}(z')\log\frac{\theta_{ab}'}{1-\theta_{ab}'}+n_{ab}(z')\log(1-\theta_{ab}'))\\
&&-\sum_{k-1\leq a\leq b\leq k}(m_{ab}\log\frac{\theta_{ab}^*}{1-\theta_{ab}^*}
+n_{ab}\log(1-\theta_{ab}^*)))-n\mu+O_p(\frac{k^2\log k}{n}))/\sigma(\theta^*)\\
&\stackrel{d}{\rightarrow}& N(0,1).
\end{array}
\]
\end{proof}

\subsection{Proof of Corollary \ref{corollary:underfit:a}}
By \citeauthor{Hoeffding:1963}'s (\citeyear{Hoeffding:1963}) inequality, we have
\[
\begin{array}{lll}
P(\max_{1\leq a\leq b\leq k}\mid \theta_{ab}^*-\hat{\theta}_{ab}\mid>t)&= &P(\max_{1\leq a\leq b\leq k}\mid \rho_n\tilde{\theta}_{ab}^*-\rho_n(\rho_n^{-1}\hat{\theta}_{ab})\mid>t)\\
& \leq & \sum_{1\leq a\leq b\leq k}P(\mid \tilde{\theta}_{ab}^*-\rho_n^{-1}\hat{\theta}_{ab}\mid>\rho_n^{-1}t)\\
& \leq & \sum_{1\leq a\leq b\leq k}e^{-2\rho_n^{-2}t^2n_an_b}\\
& \leq & e^{2\log k-2C_1^2n^2\rho_n^{-2}t^2}.
\end{array}
\]
It implies that
\[
\max_{1\leq a\leq b\leq k}\mid \theta_{ab}^*-\hat{\theta}_{ab}\mid=O_p(\frac{\rho_nk\sqrt{\log k}}{n}).
\]

By the proof of Theorem \ref{theorem:fit:a}, we have
\[
K=K_1+O_p(\frac{\rho_n^3k^3\log^{3/2}k}{n})=O_p(\rho_nk^2\log k)+O_p(\frac{\rho_n^3k^3\log^{3/2}k}{n}),
\]

By using the techniques developed in \cite{Wang:Bickel:2017}, the proof is similar to that of Theorem \ref{theorem:underfit:a}.

\subsection{Proofs for Theorem \ref{theorem:fit:a}}

We first need one useful lemma below.
\begin{lemma}
\label{lemma:fit:a}
Suppose that $A\sim P_{\theta^*, z^*}$. If $k=o(n/\log n)$, with probability tending to 1, we have
\[
\max_{z\in [k]^n}\sup_{\theta\in\Theta_{k}}\log f(A|\theta,z)=\sup_{\theta\in\Theta_{k}}\log f(A|\theta,z^*).
\]
\end{lemma}

This lemma is essentially Lemma 3 in \cite{Bickel:etal:2013}. The arguments are similar and thus omitted.

\begin{proof}[Proof of Theorem \ref{theorem:fit:a}]

By Taylor's expansion, we have
\[
\begin{array}{lll}
&& 2(\max_{z\in [k]^n}\sup_{\theta\in\Theta_k}\log f(A|\theta,z)-\log f(A|\theta^*,z^*))\\
&=& 2(\sup_{\theta\in\Theta_k}\log f(A|\theta,z^*)-\log f(A|\theta^*,z^*))\\
& = & 2\sum_{1\leq a\leq b\leq k}(m_{ab}\log\frac{\hat{\theta}_{ab}}{\theta_{ab}^*}+(n_{ab}-m_{ab})\log\frac{1-\hat{\theta}_{ab}}{1-\theta_{ab}^*})\\
& = & 2\sum_{1\leq a\leq b\leq k}(n_{ab}\hat{\theta}_{ab}\log\frac{\theta_{ab}^*+\hat{\theta}_{ab}-\theta_{ab}^*}{\theta_{ab}^*}
  +n_{ab}(1-\hat{\theta}_{ab})\log\frac{1-\theta_{ab}^*+\theta_{ab}^*-\hat{\theta}_{ab}}{1-\theta_{ab}^*})\\
& = & 2\sum_{1\leq a\leq b\leq k}(n_{ab}(\theta_{ab}^*+\Delta_{ab})(\frac{\Delta_{ab}}{\theta_{ab}^*}-\frac{\Delta_{ab}^2}{2\theta_{ab}^{*2}})\\
&&+n_{ab}(1-\theta_{ab}^*-\Delta_{ab})(\frac{-\Delta_{ab}}{1-\theta_{ab}^*}-\frac{\Delta_{ab}^2}{2(1-\theta_{ab}^*)^2})+O(n_{ab}\Delta_{ab}^3)) \\
  & = & 2\sum_{1\leq a\leq b\leq k}(n_{ab}(\Delta_{ab}+\frac{\Delta_{ab}^2}{2\theta_{ab}^*})
  +n_{ab}(-\Delta_{ab}+\frac{\Delta_{ab}^2}{2(1-\theta_{ab}^*)})+O(n_{ab}\Delta_{ab}^3)),
\end{array}
\]
where $\Delta_{ab}=\hat{\theta}_{ab}-\theta_{ab}^*$. By the proof of Theorem \ref{theorem:underfit:a}, we have
\[
\begin{array}{lll}
2L_{k,k}
& = & 2\sum_{1\leq a\leq b\leq k}(\frac{n_{ab}\Delta_{ab}^2}{2\theta_{ab}^*}+\frac{n_{ab}\Delta_{ab}^2}{2(1-\theta_{ab}^*)})+O_p(\frac{k^3\log^{3/2}k}{n})\\
& = & \sum_{1\leq a\leq b\leq k}\frac{n_{ab}(\hat{\theta}_{ab}-\theta_{ab}^*)^2}{\theta_{ab}^*(1-\theta_{ab}^*)}+O_p(\frac{k^3\log^{3/2}k}{n}),
\end{array}
\]
which converges in distribution to the Chi-square distribution with $k(k+1)/2$ degrees of freedom by the central limit theory.
\end{proof}

\subsection{Proof of Corollary \ref{corollary:fit:a}}
 Note that \[
\max_{1\leq a\leq b\leq k}\mid \theta_{ab}^*-\hat{\theta}_{ab}\mid=O_p(\frac{\rho_nk\sqrt{\log k}}{n}).
\]
By using the techniques developed in \cite{Wang:Bickel:2017}, the proof is similar to that of Theorem \ref{theorem:fit:a}.

\subsection{Proofs for Theorem \ref{theorem:overfit:a}}
The idea for the proofs is to
embed a $k$-block model in a larger model by appropriately splitting the labels $z^*$.
Define $\nu_{k'}=\{z\in[k']$: there is at most one nonzero entry in every row of $R(z,z^*)\}$. $\nu_{k'}$ is obtained by splitting of $z^*$ such that every block in $z$ is always a subset of an existing block in $z^*$. It follows from the definition of $\nu_{k'}$ there exists a surjective funciton $h: [k']\rightarrow [k]$ describing the block assignments in $R(z,z^*)$.

The following lemma will be used in the proof of Theorem \ref{theorem:overfit:a}.
\begin{lemma}
\label{lemma:overfit:a}Suppose that $A\sim P_{\theta^*, z^*}$. With probability tending to 1,
\[
\max_{z\in [k']^n}\sup_{\theta\in\Theta_{k'}}\log f(A|\theta,z)\leq \alpha n\log k'+\sup_{\theta\in\Theta_{k}}\log f(A|\theta,z^*),
\]
where $0<\alpha\leq1-\frac{C}{\log k'}+\frac{2\log n+\log k}{n\log k'}$.
\end{lemma}

\begin{proof}
The proof is similar to that of  Lemma \ref{lemma:underfit:c}. Note that in this case $G(R(z),\theta^*)$ is maximized at any $z\in \nu_{k'}$ with the value $\sum_{1\leq a\leq b\leq k}p_{ab}\gamma(\theta_{ab}^*)$. Denote the optimal $G^*=\sum_{1\leq a\leq b\leq k}p_{ab}\gamma(\theta_{ab}^*)$ and
$$I_{\delta_n}^+=\{z\in[k']^n:G(R(z),\theta^*)-G^*<-\delta_n\},$$
for $\delta_n\rightarrow 0$ slowly enough.

By Lemma \ref{lemma:underfit:b},
\[
\begin{array}{lll}
&&\mid F(m(z)/n^2,n(z)/n^2)-G(R(z),\theta^*)\mid\\
& \leq & C\sum_{1\leq a\leq b\leq k'}\mid m_{ab}(z)/n^2-(R\theta^*R^T(z))_{ab}\mid\\
& =  & O_p((\log n/n)^{1/2}),
\end{array}
\]
since $\gamma$ is Lipschitz on any interval bounded away from $0$ and $1$.

For any $z_0\in \nu_{k'}$, it is easy to see
\[
\begin{array}{lll}
&&\max_{z\in I_{\delta_n}^+}\sup_{\theta\in\Theta_{k'}}\log f(A|\theta,z) \\
& \leq & \log(\sum_{z\in I_{\delta_n}^+}\sup_{\theta\in\Theta_{k'}}f(A|\theta,z))\\
& = & \log(\sum_{z\in I_{\delta_n}^+}\sup_{\theta\in\Theta_{k'}}e^{\log f(A|\theta,z)})\\
& \leq & \log (\sup_{\theta\in\Theta_{k'}}f(A|\theta,z_0)(k'-1)^ne^{O_p(n^2(\log n/n)^{1/2})-n^2\delta_n})\\
& < & \log (\sup_{\theta\in\Theta_{k'}}f(A|\theta,z_0))\\
& = & \sup_{\theta\in\Theta_{k'}}\log f(A|\theta,z_0)\\
& = & \sup_{\theta\in\Theta_{k'}}\sum_{1\leq a\leq b\leq k}\sum_{(u,v)\in h^{-1}(a)\times h^{-1}(b)}(m_{uv}\log\theta_{uv}+(n_{uv}-m_{uv})\log(1-\theta_{uv})),
\end{array}
\]
choosing $\delta_n \rightarrow 0$ slowly enough such that $\delta_n/(\log n/n)^{1/2}\rightarrow \infty$.

Let
\[
L_{ab}=
\sum_{(u,v)\in h^{-1}(a)\times h^{-1}(b)}(m_{uv}\log\theta_{uv}+(n_{uv}-m_{uv})\log(1-\theta_{uv})+\lambda(\sum_{(u,v)\in h^{-1}(a)\times h^{-1}(b)}n_{uv}-n_{ab}).
\]
Let
\[
\frac{\partial L_{ab}}{\partial n_{uv}}
=\log(1-\theta_{uv})+\lambda
=0.
\]
This implies that for $(u,v)\in h^{-1}(a)\times h^{-1}(b)$,  $\theta_{uv}$'s are all equal. Let $\theta_{uv}=\theta_{ab}$.
Hence,
\[
\begin{array}{lll}
&&\sum_{(u,v)\in h^{-1}(a)\times h^{-1}(b)}(m_{uv}\log\theta_{uv}+(n_{uv}-m_{uv})\log(1-\theta_{uv})\\
& = & m_{ab}\log\theta_{ab}+(n_{ab}-m_{ab})\log(1-\theta_{ab}),
\end{array}
\]
where $m_{ab}=\sum_{(u,v)\in h^{-1}(a)\times h^{-1}(b)}m_{uv}$ and $n_{ab}=\sum_{(u,v)\in h^{-1}(a)\times h^{-1}(b)}n_{uv}$.

Thus,
\[
\begin{array}{lll}
&&\max_{z\in I_{\delta_n}^+}\sup_{\theta\in\Theta_{k'}}\log f(A|\theta,z)\\
& \leq & \sup_{\theta\in\Theta_{k'}}\sum_{1\leq a\leq b\leq k}\sum_{(u,v)\in h^{-1}(a)\times h^{-1}(b)}(m_{uv}\log\theta_{uv}+(n_{uv}-m_{uv})\log(1-\theta_{uv}))\\
& = &\sup_{\theta\in\Theta_{k}}\sum_{1\leq a\leq b\leq k}(m_{ab}\log\theta_{ab}+(n_{ab}-m_{ab})\log(1-\theta_{ab}))\\
& = & \sup_{\theta\in\Theta_{k}}\log f(A|\theta,z^*).
\end{array}
\]

Note that treating $R(z)$ as a vector, $\{R(z)|z\in\nu_{k'}\}$ is a subset of the union of some of the $k'-k$ faces of the polyhedron $P_R$. For every $z\notin I_{\delta_n}^+$, $z\notin\nu_{k'}$, let $z_{\bot}$ be such that $R(z_{\bot})=\min_{R(z_0):z_0\in\nu_{k'}}\parallel R(z)-R(z_0)\parallel_2$. $R(z)-R(z_{\bot})$ is perpendicular to the corresponding $k'-k$ face. This orthogonality implies the directional derivative of $G(\cdot,\theta^*)$ along the direction of $R(z)-R(z_{\bot})$ is bounded away from $0$. That is,
$$\frac{\partial G((1-\epsilon)R(z_{\bot})+\epsilon R(z),\theta^*)}{\partial\epsilon}\mid_{\epsilon=0^+}<-C$$
for some universal positive constant $C$. Similar to the proof in Lemma \ref{lemma:underfit:c},
$$\sup_{\theta\in\Theta_{k'}}\log f(A|\theta,z)-\sup_{\theta\in\Theta_{k'}}\log f(A|\theta,z_{\bot})\leq-Cmn,$$
where $\mid z-z_{\bot}\mid=m$.
For some $0<\alpha\leq1-\frac{C}{\log k'}+\frac{2\log n+\log k}{n\log k'}$, we have
\[
\begin{array}{lll}
&&\max_{z\notin I_{\delta_n}^+, z\notin \nu_{k'}}\sup_{\theta\in\Theta_{k'}}\log f(A|\theta,z)\\
& \leq & \log(\sum_{z\notin I_{\delta_n}^+, z\notin \nu_{k'}}\sup_{\theta\in\Theta_{k'}}f(A|\theta,z))\\
& \leq & \log(\sum_{z\in \nu_{k'}}\sup_{\theta\in\Theta_{k'}}f(A|\theta,z)\sum_{m=1}^n(k-1)^mn^me^{-Cnm})\\
& \leq & \log\mid\nu_{k'}\mid+\max_{z\in \nu_{k'}}\sup_{\theta\in\Theta_{k'}}\log f(A|\theta,z)+\log(\sum_{m=1}^n(k-1)^mn^me^{-Cnm})\\
& < & \log\mid\nu_{k'}\mid+\max_{z\in \nu_{k'}}\sup_{\theta\in\Theta_{k'}}\log f(A|\theta,z)+\log(n^2ke^{-Cn})\\
& \leq & n\log k'+\max_{z\in \nu_{k'}}\sup_{\theta\in\Theta_{k'}}\log f(A|\theta,z)+2\log n+\log k-Cn\\
& \leq & \alpha n\log k'+\sup_{\theta\in\Theta_k}\log f(A|\theta,z^*).
\end{array}
\]
\end{proof}

\begin{proof}[Proof of Theorem \ref{theorem:overfit:a}]
By Lemma \ref{lemma:overfit:a} and Theorem \ref{theorem:fit:a},
\[
\begin{array}{lll}
&&\max_{z\in [k']^n}\sup_{\theta\in\Theta_{k'}}\log f(A|\theta,z)-\log f(A|\theta^*,z^*)\\
&\leq & \alpha n\log k'+\sup_{\theta\in\Theta_k}\log f(A|\theta,z^*)-\log f(A|\theta^*,z^*)\\
& = & \alpha n\log k'+\frac{1}{2}\sum_{1\leq a\leq b\leq k}\frac{n_{ab}(\hat{\theta}_{ab}-\theta_{ab}^*)^2}{\theta_{ab}^*(1-\theta_{ab}^*)}+O_p(\frac{k^3\log^{3/2}k}{n})\\
& = &\alpha n\log k' + O_p(k^2\log k)
\end{array}
\]

\end{proof}

\subsection{Proof of Corollary \ref{corollary:overfit:a}}
 Note that \[
\max_{1\leq a\leq b\leq k}\mid \theta_{ab}^*-\hat{\theta}_{ab}\mid=O_p(\frac{\rho_nk\sqrt{\log k}}{n}).
\]
By using the techniques developed in \cite{Wang:Bickel:2017}, the proof is similar to that of Theorem \ref{theorem:overfit:a}.

\subsection{Proof of Theorem \ref{theorem:modelselection:a}}
Let
\[
g_n(k, \lambda, A) = \max_{z\in [k]^n}\sup_{\theta\in\Theta_{k}}\log f(A|\theta,z)-(\lambda n\log k+\frac{k(k+1)}{2}\log n),
\]
and
\[
h_n(k, \lambda, A)= \max_{z\in [k]^n}\sup_{\theta\in\Theta_{k}}\log f(A|\theta,z)-\log f(A|\theta^*,z^*)-(\lambda n\log k+\frac{k(k+1)}{2}\log n).
\]

For $k'>k$, we have
\[
\begin{array}{lll}
& &  P(\ell(k')>\ell(k))=P( g_n(k', \lambda, A) > g_n(k, \lambda, A) ) \\
& = & P(  h_n(k', \lambda, A) > h_n(k, \lambda, A) ) )  \\
& \leq &  P(\alpha n\log k'+\sup_{\theta\in\Theta_k}\log f(A|\theta,z^*)-\log f(A|\theta^*,z^*)-(\lambda n\log k'+\frac{k'(k'+1)}{2}\log n)>\\
& & \sup_{\theta\in\Theta_k} \log f(A|\theta,z^*)-\log f(A|\theta^*,z^*)-(\lambda n\log k+\frac{k(k+1)}{2}\log n)).
\end{array}
\]
By Theorem \ref{theorem:overfit:a}, for $\lambda> (\alpha \log k')/ (\log k'-\log k)$, the above probability goes to zero.

For $k'<k$, by Theorem \ref{theorem:fit:a}, we have
\[
\begin{array}{lll}
&& P(\ell(k')>\ell(k)) = P( g_n(k', \lambda, A) > g_n( k, \lambda, A) ) \\
& = & P( h_n(k', \lambda, A) > h_n(k, \lambda, A) ) \\
& = &  P( h_n(k', \lambda, A)  >
\sup_{\theta\in\Theta_k}\log f(A|\theta,z^*)-\log f(A|\theta^*,z^*)-(\lambda n\log k+\frac{k(k+1)}{2}\log n))\\
& = &  P(\max_{z\in [k']^n}\sup_{\theta\in\Theta_{k'}}\log f(A|\theta,z)-\log f(A|\theta^*,z^*)>  \lambda(n\log k'-n\log k)\\
& & ~~~~~~+(\frac{k'(k'+1)}{2}\log n-\frac{k(k+1)}{2}\log n)+\sup_{\theta\in\Theta_k}\log f(A|\theta,z^*)-\log f(A|\theta^*,z^*))\\
& = &  P(\max_{z\in [k']^n}\sup_{\theta\in\Theta_{k'}}\log f(A|\theta,z)-\log f(A|\theta^*,z^*)>  \lambda(n\log k'-n\log k)\\
& & ~~~~~~+(\frac{k'(k'+1)}{2}\log n-\frac{k(k+1)}{2}\log n)+O_p(k^2\log k))\\
& = &  P(n^{-1}(\max_{z\in [k']^n}\sup_{\theta\in\Theta_{k'}}\log f(A|\theta,z)-\log f(A|\theta^*,z^*))-n\mu\\
& & ~~~~>  -n\mu+n^{-1}(\lambda(n\log k'-n\log k)+(\frac{k'(k'+1)}{2}\log n-\frac{k(k+1)}{2}\log n))+O_p(\frac{k^2\log k}{n})).
\end{array}
\]
By Theorem \ref{theorem:underfit:a}, the above probability goes to zero by noticing that $\lambda (\log k'-\log k)$ goes to infinity at the rate of $\log k$.

\subsection{Proof of Corollary \ref{corollary:modelselection:a}}
The proof is similar to that of Theorem \ref{theorem:modelselection:a} and thus is omitted.

\subsection{Proof of Theorem \ref{theorem:extension}}
By Taylor expansion, we have
\[
\begin{array}{lll}
&& 2(\max_{z\in [k]^n}\sup_{\theta\in\Theta_k}\log f(A|\theta,\omega,z)-\log f(A|\theta^*,\omega,z^*))\\
&=& 2(\sup_{\theta\in\Theta_k}\log f(A|\theta,\omega,z^*)-\log f(A|\theta^*,\omega,z^*))\\
& = & 2\sum_{1\leq a\leq b\leq k}(m_{ab}\log\frac{\hat{\theta}_{ab}}{\theta_{ab}^*}-n_{ab}(\hat{\theta}_{ab}-\theta_{ab}^*))\\
& = & 2\sum_{1\leq a\leq b\leq k}(n_{ab}\hat{\theta}_{ab}\log\frac{\theta_{ab}^*+\hat{\theta}_{ab}-\theta_{ab}^*}{\theta_{ab}^*}
  -n_{ab}(\hat{\theta}_{ab}-\theta_{ab}^*))\\
& = & 2\sum_{1\leq a\leq b\leq k}(n_{ab}(\theta_{ab}^*+\Delta_{ab})(\frac{\Delta_{ab}}{\theta_{ab}^*}-\frac{\Delta_{ab}^2}{2\theta_{ab}^{*2}})-n_{ab}\Delta_{ab}+O(n_{ab}\Delta_{ab}^3))\\
& = & 2\sum_{1\leq a\leq b\leq k}(\frac{n_{ab}\Delta_{ab}^2}{\theta_{ab}^*}-\frac{n_{ab}\Delta_{ab}^2}{2\theta_{ab}^*}+O(n_{ab}\Delta_{ab}^3))\\
& = & \sum_{1\leq a\leq b\leq k}(\frac{n_{ab}\Delta_{ab}^2}{\theta_{ab}^*}+O(n_{ab}\Delta_{ab}^3))\\
& = & \sum_{1\leq a\leq b\leq k}\frac{n_{ab}(\hat{\theta}_{ab}-\theta_{ab}^*)^2}{\theta_{ab}^*}(1+o_p(1))\\
& = &O_p(k^2\log k)
\end{array}
\]

\section*{Acknowledgments}
We are very grateful to two referees, an AE and the Editor for their valuable comments that
have greatly improved the manuscript. Hu is partially supported by the National Natural Science Foundation of China
(Nos. 11471136, 11571133) and the Key Laboratory of Applied Statistics of MOE (KLAS) grants (Nos. 130026507, 130028612). Qin is partially supported by the National Natural Science Foundation of China (No. 11871237). Yan is partially supported by the National Natural Science Foundation of China (No. 11771171), the Fundamental Research Funds for the Central Universities and the Key Laboratory of Applied Statistics of MOE (KLAS) grants (Nos. 130026507, 130028612). Zhao is partially supported by NSF DMS 1840203.

\end{document}